\documentclass[a4paper, 10pt]{article} 
\PassOptionsToPackage{centertags}{amsmatch}
\usepackage{mathtools}
\usepackage{amsmath}
\usepackage{amsmath}
\usepackage[dutch,british]{babel}
\usepackage{amsfonts}
\usepackage{amssymb} 
\usepackage{amsthm}
\usepackage{newlfont}
\usepackage{color}
\usepackage{subfig}
 \usepackage{bbm}
 \usepackage{float}
 \usepackage{import}

\usepackage{mathrsfs}
\usepackage{pstricks,pst-node}
\usepackage{graphicx}
\usepackage{fancybox}
\usepackage{geometry}
 \usepackage{psfrag}
 
 \usepackage{multirow}

\theoremstyle{plain}
  \newtheorem{theorem}{Theorem}[section]
      \newtheorem{assumption}{Assumption}
  
  \newtheorem{proposition}[theorem]{Proposition}
  \newtheorem{lemma}[theorem]{Lemma}
  \newtheorem{remark}[theorem]{Remark}
\theoremstyle{definition}
  \newtheorem{definition}{Definition}[section]

\theoremstyle{remark}

\numberwithin{equation}{section}

\DeclareMathOperator{\Tr}{Tr}

 \DeclareMathOperator{\supp}{Supp}

\newcommand\otimesal{\mathop{\hbox{\raise 1.6 ex
  \hbox{$\scriptscriptstyle\mathrm{al}$}
\kern -0.92 em \hbox{$\otimes$}}}}
\newcommand\oplusal{\mathop{\hbox{\raise 1.6 ex
  \hbox{$\scriptscriptstyle\mathrm{al}$}
\kern -0.92 em \hbox{$\oplus$}}}}
\newcommand\Gammal{\hbox{\raise 1.7 ex
\hbox{$\scriptscriptstyle\mathrm{al}$}\kern -0.50 em $\Gamma$}}
\renewcommand\i{\mathrm{i}}

\let\al=\alpha \let\be=\beta  

  \let\ga=\gamma 
\let\ka=\kappa \let\la=\lambda \let\om=\omega

 \let\Ga=\Gamma  \let\Om=\Omega
  \let\Si=\Sigma

\newcommand{\caA}{{\mathcal A}}

\newcommand{\caC}{{\mathcal C}}
\newcommand{\caD}{{\mathcal D}}
\newcommand{\caE}{{\mathcal E}}

\newcommand{\caG}{{\mathcal G}}
\newcommand{\caH}{{\mathcal H}}

\newcommand{\caJ}{{\mathcal J}}

\newcommand{\caO}{{\mathcal O}}

\newcommand{\caS}{{\mathcal S}}
\newcommand{\caT}{{\mathcal T}}

\newcommand{\caV}{{\mathcal V}}
\newcommand{\caW}{{\mathcal W}}

\newcommand{\scrB}{{\mathscr B}}

\newcommand{\scrE}{{\mathscr E}}

\newcommand{\scrG}{{\mathscr G}}
\newcommand{\scrH}{{\mathscr H}}

\newcommand{\scrR}{{\mathscr R}}
\newcommand{\scrS}{{\mathscr S}}
\newcommand{\scrT}{{\mathscr T}}

\newcommand{\bbC}{{\mathbb C}}

\newcommand{\bbE}{{\mathbb E}}

\newcommand{\bbN}{{\mathbb N}}

\newcommand{\bbR}{{\mathbb R}}

\newcommand{\bbZ}{{\mathbb Z}}

\newcommand{\opunit}{\text{1}\kern-0.22em\text{l}}

\newcommand{\frh}{{\mathfrak h}}

\newcommand{\frp}{{\mathfrak p}}

\newcommand{\frA}{{\mathfrak A}}
\newcommand{\frB}{{\mathfrak B}}

\newcommand{\frG}{{\mathfrak G}}

\newcommand{\frS}{{\mathfrak S}}

\newcommand{\bsI}{{\boldsymbol I}}

\newcommand{\e}{{\mathrm e}}

\renewcommand{\d}{{\mathrm d}}
\newcommand{\sys}{{\mathrm S}}
\newcommand{\res}{{\mathrm F}}

\newcommand{\Dom}{\mathrm{Dom}}

\newcommand{\beq}{ \begin{equation} }
\newcommand{\eeq}{ \end{equation} }
\newcommand{\bet}{ \begin{theorem} }
\newcommand{\eet}{ \end{theorem} }

\newcommand{\initial}{0}
\newcommand{\final}{n+1}

\newcommand{\lone}{\mathbbm{1}}

\newcommand{\baq}{\begin{eqnarray}}
\newcommand{\eaq}{\end{eqnarray}}

\renewcommand{\supp}{\mathrm{Supp}}

 \newcounter{smallarabics}
\newenvironment{arabicenumerate}
{\begin{list}{{\normalfont\textrm{\arabic{smallarabics})}}}
  {\usecounter{smallarabics}\setlength{\itemindent}{0cm}
  \setlength{\leftmargin}{5ex}\setlength{\labelwidth}{4ex}
  \setlength{\topsep}{0.75\parsep}\setlength{\partopsep}{0ex}
   \setlength{\itemsep}{0ex}}}
{\end{list}}

\newcounter{smallroman}

\newcommand{\ben}{\begin{arabicenumerate}}
\newcommand{\een}{\end{arabicenumerate}}

\newcommand{\norm}{ \|}
\newcommand{\str}{ |}

\newcommand{\initialresfinite}{P_{\Om}}

\newcommand{\realinitial}{\ltimes}
\newcommand{\realfinal}{\rtimes}
\newcommand{\adjoint}{\mathrm{ad}}
\newcommand{\ad}{\adjoint}

\newcommand{\uw}{\underline{w}}

\newcommand{\dist}{d}
\newcommand{\distance}{\mathrm{dist}}

\newcommand{\weird}{\diamond}
\newcommand{\normw}{\norm_{\weird}}

\newcommand{\inter}{\mathrm{I}}

\newcommand{\sims}{\underset{\mathrm{s}}{\sim}  }
\newcommand{\opprod}{\otimes}
\newcommand{\indicator}{1}

\newcommand{\gs}{\mathrm{gs}}

\begin{document}

\begin{center}
\large{ \bf{Minimal velocity estimates and soft mode bounds for the massless spin-boson model}}
 \\
\vspace{15pt} \normalsize

{\bf   W.  De Roeck\footnote{
email: {\tt
wojciech.deroeck@itf.fys.kuleuven.be, supported by the DFG, the German Research Fund}}} \\
\vspace{10pt} 
{\it Institute for Theoretical Physics  \\
K.U.\ Leuven   \\
Celestijnenlaan  200D, B 3001  Leuven, Belgium
} \\

\vspace{15pt}

{\bf   A. Kupiainen\footnote{
email: {\tt    antti.kupiainen@helsinki.fi, supported by ERC and Academy of Finland  }}  }\\
\vspace{10pt} 
{\it   Department of Mathematics \\
 University of Helsinki \\ 
P.O. Box 68, FIN-00014,  Finland 
} \\

\vspace{15pt}

\end{center}

\begin{abstract}
We consider generalised versions of the spin-boson model at small coupling. We assume the spin (or atom) to sit at the origin $0 \in \bbR^d$ and the propagation speed $v_p$ of free bosons to be constant, i.e.\ independent of momentum. In particular, the bosons are massless. 
We prove detailed bounds on the mean number of bosons contained in the ball $\{\str x \str \leq v_p t \}$. In particular, we prove that, as $t \to \infty$, this number tends to an asymptotic value that can be naturally identified as the mean number of bosons bound to the atom in the ground state.  Physically, this means that bosons that are not bound to the atom, are travelling outwards at a speed that is not lower than $v_p$, hence the term 'minimal velocity estimate'. 
Additionally, we prove bounds on the number of emitted bosons with low momentum (soft mode bounds). This paper is an extension of our earlier work in \cite{deroeckkupiainenphotonbound}. Together with the results in \cite{deroeckkupiainenphotonbound}, the bounds of the present paper suffice to prove asymptotic completeness, as we describe in \cite{deroeckkupiainengriesemer}. 
\end{abstract}

\section{Model and result}
This paper provides technical tools to prove asymptotic completeness for some models of quantum field theory with massless bosons.   These tools complement those developed in \cite{deroeckkupiainenphotonbound} and also their proof is to a large extent parallel to the latter. Therefore, we refer the reader to \cite{deroeckkupiainenphotonbound} for an extended motivation of the model and relevant references, and to \cite{deroeckkupiainengriesemer} for a discussion of asymptotic completeness.   Suffice it so say here that interest in the rigorous theory of such models was revived by work on non-relativistic quantum electrodynamics, see e.g.\ \cite{bachfrohlichsigalqed,spohnbook}
  We first introduce the model and state the result, and then, in Section \ref{sec: discussion}, we discuss the results in this paper.

\subsection{The model}

Our model consists of a small system (atom, spin) coupled to a free bosonic field.  The Hilbert space of the total system is 
\beq
\scrH
= \scrH
_\sys \otimes \scrH
_\res
\eeq
where $\scrH
_\sys$, the atom/spin space ($\sys$ for 'small system'), is finite dimensional,  $\scrH
_\sys \sim \bbC^{d_\sys}$ for some ${d_\sys}<\infty$.  The field space $\scrH
_\res$ is the bosonic Fock space $ \Ga(\frh)$
built from the single particle space $\frh=L^2(\bbR^d)$:
\beq
\scrH
_\res = \Ga(\frh)  =\mathop{\oplus}\limits_{n=0}^\infty  P_{\mathrm{Symm}} \frh^{\otimes^{n}
} 
\eeq
with $P_{\mathrm{Symm}} $ is the projection to symmetric tensors and $  \frh^{{\otimes^0}}\equiv \bbC $.
The total Hamiltonian is of the form
\beq \label{def: tot ham}
H= H_\sys \otimes \lone+ \lone \otimes H_\res + H_{\inter}
\eeq
where
\begin{enumerate}
\item
 $H_\sys$ is a hermitian matrix acting on $\scrH
_\sys$.
\item $H_\res $ is the Hamiltonian of the free field, given by 
\beq
H_\res= \int_{\bbR^d} \d k \str k\str a^*_{k} a_{k} 
\eeq
where $a^*_{k}, a_{k} $ are the bosonic creation/annihilation operators of a Fourier (momentum) mode $k \in \bbR^d$ satisfying the `Canonical Commutation Relations'
$$
[a^{}_k,a^*_{k'}]=\delta(k-k').
$$
We also note that we have set the free propagation speed $v_p$ (see abstract) to be $1$ in choosing the `dispersion law' to be $\str k \str$ rather than $v_p \str k \str$. 
\item The coupling  $H_{\inter}$ is of the form 
\beq
H_{\inter}= \la D \otimes   \Phi(\phi) 
\eeq
where $D$ is a Hermitian matrix acting on $\scrH
_\sys$,  $\la \in \bbR$ is a coupling constant,  $\phi \in L^2(\bbR^d)$ is a ``form factor" that imposes some infrared and ultraviolet regularity in the model and $\Phi$ is the self-adjoint (Segal) field operator
\beq \label{def: segal}
 \Phi(\phi)=   a^*{(\phi)}+ a{(\phi)}, \qquad    a{(\phi)} =  \int_{\bbR^d} \d k \, \overline{ \hat\phi(k)}a_k
 \eeq
where  $\hat \phi$ denotes the Fourier transform of $\phi$. 
\end{enumerate}
If the form factor $\phi$ satisfies
\beq \label{eq: cond gs}
\int \d k \str \hat\phi(k)\str^2 (1+\frac{1}{\str k\str})  <\infty,
\eeq
then the operator $ H_\inter$ is relatively bounded w.r.t.\ $H_\res$ with arbitrarily small relative bound and therefore the Hamiltonian $H$ in \eqref{def: tot ham} is self-adjoint on the domain of $H_\res$ by the Kato-Rellich theorem. Hence the unitary dynamics $\e^{-\i t H}$ is well-defined and we  set $\Psi_t =\e^{-\i  t H} \Psi_0$ with $\Psi_0 \in \scrH
$.
A lot of work has been devoted to this model, in particular to its spectral theory, but we do not discuss this here. Instead, references are collected in \cite{deroeckkupiainenphotonbound,deroeckkupiainengriesemer}.

\subsection{Assumptions}

We describe now our assumptions on the form factor.  Its infrared
(small Fourier mode $k$) behaviour determines temporal correlations in the
model and some regularity near $k=0$ is needed.  Roughly speaking, we need
to assume
\beq \label{eq: intuitive behaviour form}
\hat\phi(k)\sim {|k|^{-\frac{d-2-\al}{2}}}
\eeq
with some $\al>0$ as $|k|\to 0$. 
\begin{definition} Let $0< \al<1$.  We define the subspace $\frh_{ \al} \subset \frh$ to consist of
 $\psi \in \frh$  such that $\hat\psi\in  C^3(\bbR^d\setminus\{0\})$, the support of $\hat\psi$ is bounded, and, for all multi indices $m$ with $|m|\leq 3$,
 \beq
\label{betadecay}
|\partial^m_k\hat\psi(k)|\leq C|k|^{(\beta-d+2)/2-|m|}.
\eeq
for  some $\be>\al$ and $C <\infty$.
\end{definition}
In the following two assumptions, we fix once and for all the form factor, the dimension $d$, and the operators $H_\sys$ and $D$. These choices and assumptions are assumed to hold throughout the article and they will not be repeated. 

The first assumption controls the infrared behaviour of the model.
\begin{assumption}[$\al$-Infrared regularity] \label{ass: infrared behaviour}
The form factor $\phi$  is in $\frh_\al$ 
 and the dimension $d\geq 3$. 
\end{assumption}
One of the most intuitive consequences of this assumption is the decay of correlations for the free boson dynamics. Indeed, by stationary phase estimates, see Appendix, we get \beq \label{eq: first correlation function}
 \int \d k \,  \str\hat\phi(k)\str^2  \e^{\i t \str k \str} = \caO( t^{-(2+\al)}), \qquad t \to \infty.
\eeq
 The expression on the left hand side  will appear in evaluating multitime correlations between the interaction terms $H_I$.
The second assumption ensures that the coupling is effective.  This assumption is very likely not necessary for our results, but it \emph{is} required for our proof.  To clarify this, consider the case $ \phi=0$, or equivalently $\la=0$, then the atom and field are not coupled and  the evolution of the field is given  by  (the quantisation of) the linear wave equation. In that case, our results can be proven by standard dispersive estimates.     However, the proofs of the results in \cite{deroeckkupiainenphotonbound} (on which the present article is  based) would break down as they rely on dissipative behaviour of the small system $\sys$, which does not occur if $\sys$ is not coupled to the field.  

 \begin{assumption}[Fermi Golden Rule] \label{ass: fermi golden rule}
We assume that  the spectrum of $H_\sys$ is non-degenerate (all eigenvalues are simple) and we let $e_0 := \min \sigma( H_\sys)$ (atomic ground state energy). Most importantly, we assume that for any  eigenvalue $e \in \sigma( H_\sys), e \neq e_0$, there is a sequence  $e(i), i=1, \ldots, n$ of eigenvalues such that 
\beq
e= e(1) >  e(2) >\ldots >  e(n)= e_0,  \qquad  \textrm{and} \qquad \forall i =1,\ldots, n-1:     j(e(i), e(i+1)) >0 
\eeq
with  $j(\cdot,\cdot)$ given by
\beq
j(e,e') :=   2 \pi  \Tr[P_e DP_{e'}DP_e]  \,  \int_{\bbR^d}  \d k  \, \delta(\str k \str- (e-e'))\str \hat \phi(k)\str^2  \label{eqexpression jump rates}
\eeq
 where $P_e$ is the spectral projector corresponding to the eigenvalue $e$ and the right hand side is well-defined since $\hat \phi$ is continuous away from $0$. 
\end{assumption}

This assumption will not enter explicitly in the present article, but it is necessary for the crucial Lemma \ref{lem: spectral gap}, whose proof is in \cite{deroeckkupiainenphotonbound}.

\subsection{Results}

We now state our main results.  To choose appropriate initial states, we introduce the unitary \emph{Weyl operator} $\caW(\psi), \psi \in \frh$
\beq 
\caW(\psi) = \e^{\i \Phi(\psi)}
\eeq
with the (Segal) field operator $\Phi(\psi)$ as in \eqref{def: segal}, and define the dense subspace
\beq 
\caD_\al :=  \mathrm{Span} \{ \psi_\sys \otimes \caW(\psi) \Om \,  \str\,  \psi \in \frh_{\al}, \psi_\sys \in \scrH_\sys \}
\eeq 
with $\Om$ the normalised vacuum vector.
The density of $\caD_\al $ in $\caH$ follows from the density of $\frh_{\al}$ in $\frh$. We will choose the initial vector $\Psi_0 \in \caD_\al$ with $\norm\Psi_0\norm=1$ and we write $\Psi_t = \e^{-\i t H} \Psi_0$.

Fix a  $C^{\infty}$  function $\theta: \bbR^d \to [0,1]$ with compact
support in the ball centered at origin with radius 
$r_\theta<1$.  Since $\theta$ will be used to localise both in real $x-$space and
in Fourier $k-$space we use the notation $\theta(x)$ for the multiplication operator
and $\theta(k)$ for the Fourier multiplier.   To any self-adjoint operator $b$ on $\frh$, we associate its second quantisation $\d\Gamma(b)$, a self-adjoint operator on $\Ga(\frh)$, and we also write $\d\Gamma(b)$ for $\lone \otimes \d\Gamma(b)$, acting on $\scrH$. 

In the statement of our theorems, $\breve C$ denotes constants that depend on $\Psi_0,\theta,\al$, the dimension $d, d_\sys$, and the parameters of the Hamiltonian \eqref{def: tot ham}, i.e.\ the form factor $\phi$ and the operators $H_\sys,D$, but not on $\la$.  We recall that we always assume Assumptions \ref{ass: infrared behaviour} and \ref{ass: fermi golden rule} to hold. 
\bet[Soft mode bound]  \label{thm: soft photon bound} 
There exists $\lambda_0>0$ such that for all $\la$ with $0< \str \la \str \leq \la_0$, 
\beq \label{eq: statement soft photon bound}
\sup_{t \geq 0} \left\str\langle \Psi_t,   \d \Gamma(\theta( k/\delta )) \Psi_t \rangle \right\str \leq   \breve C  \delta^{\al/2}
\eeq
for any 
$\delta> 0$, smooth indicator $\theta$ as above, and $\Psi_0 \in \caD_\al$.   
\eet
This result complements the boson number bound in \cite{deroeckkupiainenphotonbound}:
 \beq \label{eq: first photon bound}
\sup_{t \geq 0} \left\str\langle \Psi_t,   N \Psi_t \rangle \right\str \leq    \breve C
\eeq
with $N=\d\Gamma(\lone)$ the number operator. Although  the infrared condition in that paper is slightly different, an obvious application of Lemma \ref{lem: decay} in Appendix \ref{app: propagation estimates}  allows to derive that condition from our present infrared condition, i.e.\ from Assumption \ref{ass: infrared behaviour}, such that \eqref{eq: first photon bound} holds in the present framework as well. 
 Inspecting the proof of Theorem \ref{thm: soft photon bound}, we see that the bound $\breve C \delta^{\al/2}$ in \eqref{eq: statement soft photon bound} can be replaced by $\breve C(\al')\delta^{\al'}$, for any $\al'<\al$, at the cost of making the constant  dependent on $\al'$.  

\bet[Minimal velocity estimate] \label{thm: propagation estimate}
Let $0< \str\la\str \leq \la_0$ as in Theorem \ref{thm: soft photon bound} and fix  an initial state vector $\Psi_0 \in \caD_\al$ with $\norm \Psi_0\norm=1$ and a smooth indicator $\theta$ as in Theorem \ref{thm: soft photon bound}. 
For any `cutoff time' $t_c \geq \str\la\str^{-2}$,  the limit
\beq
a(t_c, \theta):= \lim_{t\to \infty} \langle \Psi_t,    \d \Gamma(\theta( x/{t_c} )) \Psi_t \rangle
\eeq
exists and
\beq
\left\str\langle \Psi_t,    \d \Gamma(\theta( x/{t_c} )) \Psi_t \rangle -  a(t_c,\theta)  \right \str   \leq   \breve C   (1+ t )^{-\al},
\eeq
uniformly in $t_c$ for $t \geq t_c \geq \str\la\str^{-2} $.
\eet 
The restriction $t_c \geq \str\la\str^{-2}$ is not necessary for the result to hold, but its elimination requires an additional step in our proof, and therefore we avoided it, since we are mainly interested in the case  $t_c=t$, see below.
The obvious interpretation of  Theorem \ref{thm: propagation estimate} is that 
\beq \label{eq: obvious interpret}
a(t_c,\theta) =
\langle \Psi_{\gs},  \d \Gamma(\theta( x/{t_c} )) \Psi_{\gs} \rangle
\eeq
where $\Psi_{\gs}$ is the unique (up to a phase) normalised ground state of $H$, that can indeed be proven to exist given our assumptions, see \cite{deroeckkupiainenphotonbound}. This interpretation is correct but we postpone its statement to \cite{deroeckkupiainengriesemer} because the identification of the limit requires somehow different reasoning that does not naturally fit into the present paper. 
The most natural and, as far as we see, useful form of this result is obtained if we assume \eqref{eq: obvious interpret} and take $t_c=t$. Then the resulting statement is 
\beq
\left\str\langle \Psi_t,    \d \Gamma(\theta( x/t )) \Psi_t \rangle -  \langle \Psi_{\gs},  \d \Gamma(\theta( x/t )) \Psi_{\gs} \rangle \right\str   \leq   \breve C   (1+ t )^{-\al}
\eeq
which is a key ingredient in \cite{deroeckkupiainengriesemer}.  This is also the claim that was announced in the abstract.

\subsection{Discussion}\label{sec: discussion}

In \cite{deroeckkupiainenphotonbound}, we established two results. On the one side,  we showed that for localised observables $O$, i.e.\ those concerning the atom and the field in the neighbourhood of the atom,  the expectation value $\langle \Psi_t, O \Psi_t \rangle$ converges to the stationary value $\langle \Psi_\gs, O \Psi_\gs \rangle$ (assuming that a ground state $\Psi_\gs$ exists).   On the other side, we showed that the number of emitted bosons is bounded independently in time, i.e.\ \eqref{eq: first photon bound}.  Intuition suggests that the emitted bosons behave as free bosons once they are sufficiently far from the atom. 
One consequence of this intuition is that the number of bosons in a spatial region of the form 
 $$ c_1t <  \str x \str < c_2 t, \qquad \text{with}\,\, 0 < c_1 <c_2 <1,$$ 
 should tend to $0$, as $t \to \infty$, where we recall that we have set the propagation speed of free bosons to be  $1$. 
 In case $c_1=0$, this is not quite true since some bosons are bound by the atom in the interacting ground state, but in that case it is still true that the expectation value of the number of bosons tends to a constant value as $t \to \infty$. 
This result is achieved in Theorem \ref{thm: propagation estimate} (up to some issues pointed out below the statement of this theorem). 
We refer to it as a minimal velocity estimate since it excludes the existence of bosons travelling with a speed lower than the propagation speed of free bosons.  

Once one knows that the number of emitted bosons remains finite, such minimal velocity estimates can be obtained by operator techniques as well, but we prefer to modify slightly the  polymer expansion in \cite{deroeckkupiainenphotonbound} to obtain these results.  The approach via operator techniques has been explored in \cite{faupin2012quantum}.   Both in our work, and in \cite{faupin2012quantum} the motivation comes from the fact that minimal velocity estimates are helpful in proving asymptotic completeness.  However, in the present article, our aim is also to illustrate that the `polymer expansion'-approach to problems in open quantum systems, that we started in \cite{deroeckkupiainen}, can be adapted to a variety of problems. 

A second result in this paper is the soft mode bound, Theorem \ref{thm: soft photon bound}. This result could be obtained completely analogously to the treatment of \cite{deroeckkupiainenphotonbound},  since the one-particle operator $b=\theta(k/\delta)$ is invariant under the free boson dynamics, but to make the  present paper more streamlined, we treat it analogously to Theorem \ref{thm: propagation estimate}, which concerns a non-invariant $b$-operator. 
Note that in \cite{faupinsigalcommentnumber}, such an analogy to \cite{deroeckkupiainenphotonbound} is used to control $\d \Ga(b)$ for $b=1/\str k \str$, which is of course also invariant.  

\subsection{Strategy of the proof and outline of the paper}
Since the strategy of this paper is so intimately connected to \cite{deroeckkupiainenphotonbound}, we restrict ourselves here to a rough outline of the proof. In particular, the arguments in Section \ref{sec: overview polymers} are analogous to \cite{deroeckkupiainenphotonbound} and they are more thoroughly explained there. 

\subsubsection{Polymer representation}\label{sec: overview polymers}
We set out to control the quantity $\langle \Psi_t, \d \Ga(b) \Psi_t\rangle$ for $t=n/\la^2$ and with the one-boson operator $b$ being either $ \theta(x/t_c)$ or $\theta(k/\delta)$.
We first construct a \emph{polymer representation} of this quantity:
\beq \label{eq: intro polymer}
\langle \Psi_t, \d \Ga(b) \Psi_t\rangle =  \sum_{\caA} \sum_{A' \in \caA}  v(A' \cup \{n+1\})  \prod_{A \in \caA, A\neq A' } v(A)  
\eeq
where $A',A$  range over nonempty subsets of $\{0,1,\ldots, n\}$ (\emph{polymers}) and $\caA$ ranges over collections of polymers that are pairwise disjoint and  non-adjacent.  The initial state $\Psi_0$ influences the \emph{polymer weights} $v(A)$ for $A \ni 0$ and the observable  $\d \Ga(b) $ influences $v(A)$ for $ A \ni (n+1)$. 
For the weights of  \emph{bulk} polymers $A$, (i.e.\ not containing $0$ nor $n+1$) we need a bound that was already stated in \cite{deroeckkupiainenphotonbound}, and that, for small $\str A \str$, can be thought of as
\beq \label{eq: example bound}
\str v(A)\str \leq  \la^2  (\max A-\min A)^{-(2+\al)}
\eeq
 where the decay factor $(\max A-\min A)^{-(2+\al)}$ reflects the decay of correlations of the free field exhibited in \eqref{eq: first correlation function} and the factor $\la^2$ indicates that these polymers capture the effect of interactions. 
 
 If we replace $\d \Ga(b)$ by $\lone$, the corresponding expansion reads
 \beq \label{eq: intro polymer one}
1= \langle \Psi_t, \lone \Psi_t\rangle = \sum_{\caA}  \prod_{A \in \caA} v(A),
\eeq
i.e.\ polymers with $(n+1) \in A$ are absent. 
These polymer representations are derived in Sections \ref{sec: operator valued polymers} and \ref{sec: scalar polymer model} via an intermediate polymer representation with operator-valued polymer weights. This derivation is a purely algebraic exercise.   Then we derive bounds on the polymer weights, like \eqref{eq: example bound}. We again use the operator-valued polymer weights as a useful intermediary step. This is done in Sections \ref
{sec: estimates} and  \ref{sec: bounds on operator valued polymers}.   The main idea of these bounds is to recognise in the definition of $v(A)$ a connected graph whose vertices, roughly speaking, coincide with the elements of $A$ and whose edges $\{\tau,\tau'\}, \tau,\tau' \in A$ carry a decay factor $\str \tau'-\tau \str^{-(2+\al)}$.   Since the graph is connected, we can extract the overall decay factor $(\max A-\min A)^{-(2+\al)}$.  This description is oversimplified; in reality, some of the vertices of the graph are subsets of $A$ themselves and the decay factors are encoded into them.  The sum over graphs is performed with the help of combinatorial techniques from cluster expansions. 

\subsubsection{The weights $v(A \cup \{n+1\})$}
Let us now describe the basic intuition for the weights $v(A \cup \{n+1\})$. 
First,  by brutal approximation, we could guess that \eqref{eq: intro polymer} divided by \eqref{eq: intro polymer one} is approximated as
\beq \label{eq: intro polymer bar}
\langle \Psi_t, \d \Ga(b) \Psi_t\rangle \approx \sum_{A}  v(A \cup \{n+1\}).
\eeq
We will make this relation into an equality by replacing the weights $v$ by slightly modified weights $\bar v$.  This type of arguments are presented in Section \ref{sec: general considerations} and they are again based on cluster expansions. 
Next, let us naively expand $\langle \Psi_t, \d \Ga(b) \Psi_t\rangle$ in powers of $\la$, hence of $H_I$, up to second order.  We write $H=H_0+H_I$ and we use
\beq
\e^{\i t H_0} \d \Ga(b) \e^{-\i t H_0} =    \e^{\i t H_\res} \d \Ga(b) \e^{-\i t H_\res}  =    \d \Ga(b(t)) 
\eeq
where $b(t)=\e^{\i t \om}b \e^{-\i t \om} $ and $\om$ is the self-adjoint multiplication operator with the Fourier multiplier $\str k \str$, i.e. $\widehat{(\om\psi)} (k)= \str k \str \hat \psi(k)$. By (formal) Duhamel expansion of $\e^{\i t H}, \e^{-\i t H}$,
\begin{align}
\e^{\i t H} \d \Ga(b) \e^{-\i t H}  + \caO(\str\la\str^3)   &= \d \Ga(b(t))  \nonumber  \\[2mm]
 &+ \i \int_0^t  \d t_1 \,   H_I(t_1)  \d \Ga(b(t))  +  \mathrm{h.c.\ }     \nonumber  \\[2mm]
  &+ \int_0^t  \d t_1\int_0^t  \d t_2  \,     H_I(t_1)  \d \Ga(b(t))  H_I(t_2)    \nonumber \\[2mm]
    & - \int_0^t  \d t_1\int_{t_1}^t  \d t_2   \,    H_I(t_1) H_I(t_2)  \d \Ga(b(t))      +  \mathrm{h.c.\ }     \nonumber  
\end{align}
where we write $H_I(t)= \e^{\i t H_0} H_I \e^{-\i t H_0}  $ and (below) $D(t)= \e^{\i t H_0} D\e^{-\i t H_0}$. 
Let us now for simplicity choose  $\Psi_0= \psi_\sys\otimes \Om$ and put the above equation between $\langle \Psi_0, \,\cdot \,\Psi_0\rangle$. Then all terms in the above expansion vanish except the third line, the integrand of which can be recast as
\beq
\langle \Psi_0, H_I(t_1)  \d \Ga(b(t))  H_I(t_2) \Psi_0 \rangle  =    \langle \psi_\sys,   D(t_1)D(t_2)  \psi_\sys \rangle^{}_{\scrH_\sys} \, \times \,       \la^2 \langle  a^*(\e^{\i t_1 \om}\phi) \Om, \d \Ga(b(t))  a^*(\e^{\i t_2 \om}\phi) \Om\rangle.       \nonumber 
\eeq
The first factor on the left hand side is quasiperiodic in $t_2-t_1$ and as such it is irrelevant. The second factor  can be rewritten as 
\beq \label{eq: micro h}
\la^2 \langle  \e^{\i t_1 \om}\phi, b(t) \e^{\i t_2 \om}\phi\rangle, \eeq i.e.\ an expression in the one-boson space $\frh$: 
we get the crude cartoon
\beq \label{eq: cartoon}
\langle \Psi_t, \d \Ga(b) \Psi_t\rangle  \approx   \la^2\int_0^t  \d t_1\int_0^t  \d t_2  \,  \langle  \e^{\i t_1 \om}\phi, b(t) \e^{\i t_2 \om}\phi\rangle.
\eeq
Comparing this integral with the sum in \eqref{eq: intro polymer bar}, it is plausible that \eqref{eq: micro h} is similar to
\beq
 v(\{\tau_1,\tau_2, n+1\}), \qquad  \text{with $\tau_1, \tau_2$ such that $\tau_1\approx  \la^2t_1 $ and $\tau_2\approx  \la^2t_2 $},
\eeq
which is not quite true, but good enough for the picture that we are developing here.  The weights $ v(A \cup \{n+1\})$ with $\str A \str >2$ yield small corrections that we do not describe here. 

One of the most relevant properties of the function \eqref{eq: micro h} is that it retains the decay in $t_2-t_1$ exhibited in \eqref{eq: first correlation function}. In the case  $b= \theta(x/t_c)$, it is quite intricate to prove this uniformly in $t_c$, but this has little to do with our main technical work. Therefore, estimates of this kind are gathered in the Appendix, see in particular Lemma \ref{lem: bounds g}.  
 
\subsubsection{The long-time limit}
If one accepts the heuristic outline above, then one gets from \eqref{eq: cartoon} by the change of variables $s_i=t-t_i$, and abbreviating $M_b(s_1,s_2):=  \la^2 \langle  \e^{-\i s_1 \om}\phi, b\,  \e^{-\i s_2 \om}\phi\rangle $
\beq
\lim_{t\to\infty} \langle \Psi_t, \d \Ga(b) \Psi_t\rangle \approx    \mathop{\int}\limits_{(\bbR_+)^2}  \d s_1  \d s_2  \, M_b(s_1,s_2),
\eeq
\beq
\langle \Psi_t, \d \Ga(b) \Psi_t\rangle-\lim_{t\to\infty} \langle \Psi_t, \d \Ga(b) \Psi_t\rangle \approx  - \mathop{\int}\limits_{(\bbR_+)^2 \setminus [0,t]^2 }  \d s_1  \d s_2  \, M_b(s_1,s_2)\eeq
provided the  improper integrals on the right hand side are defined by some appropriate regularisation procedure.   If these were exact equalities, then our theorems would reduce to statements about one-boson dynamics.  Those statements can then be checked by the estimates in the Appendix. Note that for $b=\theta(k/\delta)$, $M_b(s_1,s_2)$ is a function of $s_2-s_1$ only, but its integral over $\bbR$ vanishes because of \eqref{eq: intuitive behaviour form} and the above expressions are in fact finite. For $b= \theta(x/t_c)$, $M_b(s_1,s_2)$ decays as soon as $s_1,s_2$ are larger than $t_c$ and also as $s_2-s_1\to \infty$; these are dispersive properties of the linear wave equation. 

     In Sections \ref{sec: propagation bound} and  \ref{sec: soft photon bound}, we find the full-blown version of this argument, where the right hand sides of the above equations are replaced by sums over $\bar v(A \cup \{n+1\})$. 

\subsection{Notation}
\subsubsection{Combinatorics}
We write $\bbN=\{0,1,2,\ldots\}$. 
For $\tau,\tau' \in \bbN$, $\tau<\tau'$, we define the discrete intervals
\beq
I_{\tau,\tau'} := \{\tau, \tau+1, \ldots, \tau' \}
\eeq
and $\frB_{\tau,\tau'}$ the set of collections of nonempty subsets of $I_{\tau,\tau'}$.   A relevant subset of  $\frB_{\tau,\tau'}$ is, for $j \in \bbN$, 
\beq
\frB^j_{\tau,\tau'} :=  \{ \caA \in \frB_{\tau,\tau'} \, \str \, \forall A,A' \in \caA:  (A \neq A' \Rightarrow  \distance(A,A') > j)  \}
\eeq
where $\distance(A,A'):= \min_{\tau \in A,\tau' \in A'} \str \tau-\tau'\str$ (hence  $\distance(A,A')= \infty$ if $A$ or $A'$ is empty). 
For a collection $\caA$, we set
\beq
\supp \caA := \cup_{A \in \caA}  A 
\eeq
and we need also the diameter of (finite) subsets of $\bbN$;
\beq
\dist(A) = \max A -\min A+1,\qquad  \dist(\caA) := \dist(\supp \caA). 
\eeq 
\subsubsection{Hilbert and Banach spaces}
For a Banach space $\scrE$, we let $\scrB(\scrE)$ stand for the set of bounded operators. If $\scrE$ is a Hilbert space, we will additionally use the space of trace class operators
\beq
\scrB_p(\scrE) = \{ O \in \scrB(\scrE) \, \str \, \norm O \norm_p < \infty \}
\eeq
with 
\beq
\norm O \norm_p = (\Tr  \str OO^*\str^{p/2})^{1/p}.
\eeq
For the scalar product on a Hilbert space $\scrE$, we use the notation  $\langle \psi, \psi'\rangle_{\scrE}$, often abbreviated as $\langle \psi, \psi'\rangle$.
A positive operator $\rho \in \scrB_1(\scrE)$ with $\Tr \rho=1$ is called a density matrix. 
We also use the function
\beq
\langle x \rangle := \sqrt{x^2+1},
\eeq
for real numbers and self-adjoint operators.
\subsubsection{Constants}

We denote by $c,C$ constants that depend only on the dimensions $d,d_\sys$, the parameters of the Hamiltonian \ref{def: tot ham} and the parameter $\al$, but not on $\la$. 
The precise value of these constants can be different in different equations. 
Quantities that additionally depend on the initial condition $\Psi_0$ and the smooth function $\theta$ (but not on $\la$) are denoted by $\breve c, \breve C$.

\section{Polymer Representation} \label{sec: polymer rep}

In this section, we complete the first important step of our proof, namely we rewrite all quantities of interest through a polymer representation.   This part of the paper is almost identical to a corresponding part in \cite{deroeckkupiainenphotonbound}. 

We discretise time by introducing a "mesoscopic"  time scale $\la^{-2}$. That is, we consider times  of the form $t= n/\la^2 $ with $n \in \bbN$. The discretisation will be  removed at the end of the argument.  We study 
\beq
Z_n(O, \rho_0): =  \Tr\left[ O   \e^{-\i t H }    \rho_0     \e^{\i t H}  \, \right], \qquad t= n/\la^2,  
    \label{eq: def z}
\eeq
with  the initial density matrix
$$\rho_0= \rho_{\sys,0} \otimes \caW(\psi_\realinitial) P_{\Om} \caW^*(\psi_\realinitial)
$$ 
for some density matrix 
$\rho_{\sys,0} \in \scrB_1(\scrH_\sys)$,  $\psi_\realinitial \in \frh_\al$,  and $P_{\Om}$ the one-dimensional projector on the range of the vacuum vector $\Om$. The observable $O$  is one of the following
\begin{enumerate}
\item $O= \d \Gamma (b)$ with $b=b_x:= \theta(x/t_c)$  or  $b=b_k:= \theta(k/\delta)$.  We choose $t_c=\la^{-2} n_c$ for some $n_c \in \bbN, n_c \leq n$. All estimates will be uniform in $n_c$.   Since $O$ is unbounded, the expression \eqref{eq: def z} is in need of justification, which is provided in \cite{deroeckkupiainenphotonbound}. A posteriori, one can also appeal to the convergent expansions developed in Section \ref{sec: estimates} where we construct a convergent expansion for  \eqref{eq: def z}. 
\item $O= \lone$. This case is mainly included for comparison. By cyclicity of the trace, we have $Z_n(\lone, \rho_0)=1$. 
\end{enumerate}

In most intermediary steps  of our analysis we will perform a partial trace over the field, thereby defining the reduced dynamics
\beq\label{eq: reduced dynamics}
{Q}_{n}  \rho_{\sys,0}      :=    \Tr_{\res} \left[      \e^{-\i (n/\la^2)  L}  \rho_0   \right]
\eeq
where we introduced the Liouvillian $L=\adjoint(H)$, an unbounded operator on $\scrB_1(\scrH)$.
 Sometimes, we want to incorporate the observable into the reduced analysis as well. In that case, we write
\beq\label{modifieddynamics}
{Q}_{n \str b}  \rho_{\sys,0}      :=    \Tr_{\res} \left[   \d\Gamma(b)  \e^{-\i (n/\la^2)  L} (\rho_{\sys,0} \otimes P_{\Om} )    \right].
\eeq
Here the notation differs slightly from the one in \cite{deroeckkupiainenphotonbound} where the latter object was called $\breve Q_n$ and the notation $Q_n$ was reserved for \eqref{eq: reduced dynamics} with $\psi_\realinitial=0$.  Obviously, we have
\beq
Z_n(\lone,\rho_0) = \Tr {Q}_{n}  \rho_{\sys,0}=1, \qquad   Z_n(\d\Ga(b),\rho_0) = \Tr {Q}_{n \str b}  \rho_{\sys,0}.
\eeq
The main goal of the first part of the present chapter is to find a convenient representation for $Q_n$
and $ {Q}_{n \str b}$. 
The first step is to write the evolution operators as a product where each factor corresponds to a 'mesoscopic' time slice of length $\la^{-2}$. 
With this in mind, we introduce 
\beq
U_{\tau} :  \scrB_1(\scrH) \to   \scrB_1(\scrH) 
\eeq
with
\beq  \label{def: u tau}
U_{\tau} :=      \e^{\i (\tau/\la^2)  L_\res} \e^{-\i (1/\la^2)  L} \e^{-\i ((\tau-1)/\la^2)  L_\res}   ,\qquad \tau \in I_{1,n},
\eeq
where we used the field Liouvillian $L_\res :=\adjoint(H_\res)$, 
and
\begin{align}
 U_{\initial}  \rho & :=    {\caW}(\psi_{\realinitial}) \rho \caW^*(\psi_{\realinitial}),  \label{def: u tau initial}  \\[1mm]  
 U_{\final} \rho & := \d \Gamma(b(n/\la^2))  \rho,   \label{def: u final}
\end{align}
where we wrote for brevity $\caW(\psi)$ instead of $\lone \otimes \caW(\psi)$ and $b(t)= \e^{\i t \om} b\e^{-\i t \om}$
An immediate consequence of these definitions, using cyclicity of the trace, is
\begin{align}
{Q}_{n}  \rho_{\sys,0}    &  :=    \Tr_{\res} \left[      U_n \ldots U_1U_0( \rho_{\sys,0} \otimes P_\Om)   \right],  \label{reddyn} \\[1mm]
{Q}_{n \str b}  \rho_{\sys,0}    &  :=    \Tr_{\res} \left[     U_{n+1} U_n \ldots U_1U_0( \rho_{\sys,0} \otimes P_\Om)   \right].\label{reddyn1}
\end{align}
Finally, we define the reduced dynamics 
$$
T: \scrB_1(\scrH_\sys)\to \scrB_1(\scrH_\sys)
$$
 for (mesoscopic) time $1$, starting from a product state
\beq \label{def: t}
T     \rho_{\sys,0}:=   \Tr_\res\left[  \e^{-\i (1/\la^2) L }    (\rho_{\sys,0}\otimes P_\Om)    \, \right].
\eeq
We  set $T_\tau:=T$ for $\tau=1,\ldots,n$  and 
\baq
 T_{\initial}&:=&  \left\langle    {\caW}(\psi_{\realinitial}) \Om, {\caW}(\psi_{\realinitial}) \Om \right\rangle  \lone = \lone, \label{Tinit}\\[1mm]
  T_{\final}&:=& \left\langle   \Om,  \d \Gamma(b)   \Om\right\rangle  \lone =0  \label{Tfinal}
 \eaq
 where we used that $\norm {\caW}(\psi_{\realinitial}) \Om \norm =1$ because $ {\caW}(\psi_{\realinitial})$ is unitary. The motivation for this definition will become obvious in the next section.  
 Finally, we set
 \beq \label{eq: def bigb}
  B_{\tau} :=     U_{\tau} -T_{\tau}, \qquad  \tau =0,\ldots, n+1.
\eeq
Note that $U_{\tau}$ depends on the total macroscopic time $n$ because of \eqref{def: u final}.  It is sometimes convenient, see e.g.\ Section \ref{sec: symmetry prop}, to indicate the $n$ dependence explicitly by writing $U_{\tau,n}, B_{\tau,n}$ with $ \tau \leq n+1$, such that $U_{n+1,n}$ is defined by \eqref{def: u final} and $U_{\tau,n}, \tau \leq n$ is defined by (\ref{def: u tau},\ref{def: u tau initial}).

 The next section proposes a framework whose purpose is to write an expansion for $Q_n$ and ${Q}_{n \str b}$ in which the leading terms, in a precise sense, are 
  $ T_n \ldots T_2T_1T_0 = T^n T_0$ and   $T_{n+1} T^n T_0=0$, respectively.

\subsection{Operator-valued polymers weights}   \label{sec: operator valued polymers}

\subsubsection{Operator correlation functions} \label{sec: operator correlation functions}
We abbreviate
\beq
\scrR_\sys =  \scrB(\scrB_1(\scrH_S)), \qquad    \scrR_\res =  \scrB(\scrB_1(\scrH_\res)).
\eeq
Define, for $W,W' \in \scrR_\sys  \otimes  \scrR_\res$ the object
$$W\otimes_\sys  W'\in \scrR_\sys  \otimes  \scrR_\sys \otimes \scrR_\res $$
 as an operator product in $\res$-part and tensor
 product in $\sys$-part. Concretely, let $W=W_\sys \otimes W_\res$ and $W'=W'_\sys \otimes W'_\res$.
 Then 
$$W\otimes_\sys W':= W_\sys  \otimes W'_\sys  \otimes W_\res W'_\res.$$
and we extend this by linearity to arbitrary $W, W'$. 
Iterating this construction we define for $W_i \in \scrR_\sys  \otimes  \scrR_\res  $, $i=1,\dots,m$
$$ W_m\otimes_\sys \dots\ldots \otimes_\sys W_2 \otimes_\sys W_1 \in (\scrR_\caS)^{ \otimes^m}\otimes \scrR_\res .$$
Since $\scrR_\sys$ is finite-dimensional these products are unambiguously defined.

We define the `expectation' 
$$\bbE:(\scrR_\sys)^{ \otimes^m}\otimes  \, \scrR_\res
\,\,\rightarrow \,\,(\scrR_\sys)^{ \otimes^m}$$ as
$$
\bbE (W) J:={\Tr}_{\res}  [W (J \otimes P_{\Om})], \qquad  J \in  (\scrB_1(\scrH_\sys))^{ \otimes^m}.
$$
Obviously, the action of $\bbE$ is extended to unbounded $W$ satisfying $W((\scrB_1(\scrH_\sys))^{ \otimes^m}  \otimes P_{\Om}) \in \scrB_1(\scrH_\sys^{ \otimes^m}  \otimes \scrH_\res)$.  An important example, with $m=1$, is  $T = \bbE(U_{\tau})$.   

Let $A=\{\tau_1,\tau_2,\dots, \tau_m\}\subset I_{0,n+1}$ with the convention that
$\tau_i<\tau_{i+1}$ and define the `time-ordered {correlation function}'
\beq
G_A :=    \bbE \left( B_{\tau_m}\otimes_{\sys}  B_{\tau_{n-1}}\otimes_{\sys}\dots\otimes_{\sys} B_{\tau_1} \right)   \in   (\scrR_S)^{ \otimes^m}.   \label{eqfirst definition of correlation function}
\eeq
Note that $G_A  =0$ when the set $A$ is a singleton. Indeed, since $B_\tau = U_\tau- \bbE(U_\tau)$.  
we get  $\bbE(B_{\tau})=0$.

It will be convenient to label the $\scrR_\sys$'s and to drop the subscript $\sys$ (since we will rarely need $\scrR_\res$), writing simply $\scrR$ for $\scrR_\sys$.    Let us denote by $\scrR^{\otimes^{\bbN}}$ the linear space spanned by simple tensors $   \ldots \otimes V_2\otimes V_1$ where all but a finite number of $V_j$ are equal to the identity $\lone$.  For finite subsets $A \subset \bbN$, we then define $\scrR_A$ as the finite-dimensional subspace of $\scrR^{\otimes^\bbN}$ spanned by  simple tensors $   \ldots \otimes V_2\otimes V_1$ with $V_j=\lone, j \notin A$ and we write in particular $\scrR_{\tau}=\scrR_{\{\tau\}}$.   Let $A=\{\tau_1, \tau_2, \ldots, \tau_m\}$ with $ \tau_1 < \tau_2 <\ldots < \tau_m$. 
Obviously,  $\scrR_A$ is naturally isomorphic to $\scrR^{\otimes^m}$ by identifying the right-most tensor factor to $\scrR_{\tau_1}$, the next one to $\scrR_{\tau_2}$, etc\ldots We denote this isomorphism from $\scrR^{\otimes^m}$ to $\scrR_A$ by $\bsI_A$ and we will from now on write $G_A$ to denote $\bsI_A[G_A] \in \scrR_A$ since $G_A$ acting on the `unlabelled' space $\scrR^{\otimes^m}$ will not be used, except briefly in the upcoming Section \ref{sec: symmetry prop}

Consider a collection $\caA$ of disjoint subsets of $\bbN$, then each of the spaces $\scrR_{A \in \caA}$ is a subspace of $\scrR_{\supp \caA}$ where $\supp \caA =\cup_{A \in \caA} A$. Given a  collection of operators $K_A \in \scrR_A$, we have $\prod_{A} K_A  \in \scrR_{\supp \caA}$. 
However, we prefer to denote  such products  by
\beq
\mathop{\otimes}\limits_{A \in \caA}   K_A      \in \scrR_{\supp \caA},
\eeq
i.e.\ we keep the tensor product explicit in the notation.

\subsubsection{Symmetry properties}\label{sec: symmetry prop}

For later use, we also establish some symmetry properties of the operators $G_A$. To do this, it is more natural to keep the definition \eqref{eqfirst definition of correlation function}, i.e.\ to view $G_A$ as an element of $\scrR^{{\otimes^m}}$ with $m=\str A\str$ instead of $\scrR_A$.   As announced following equation \eqref{eq: def bigb}, we write $B_{\tau,n}$ instead of $B_{\tau}$ and then also $G_{A,n}$ instead of $G_A$ to indicate the dependence on the final time $n$:

Let $\tau \in \bbN$, $n'>n$ and $A$ such that  both $A,A+\tau$  are subsets of $ I_{1,n}$, then
\beq
G_{A+\tau,n} = G_{A,n} =  G_{A,n'}
\eeq
 because $ \e^{-\i t L_\res} \initialresfinite=\initialresfinite$. 
Similarly, if $b=b_k$, then, with $A,\tau,n'$ as above 
\beq
G_{(A+\tau)\cup \{n+1\},n} = G_{A\cup \{n+1\},n} =  G_{A\cup \{n'+1\},n'}.
\eeq
This follows by the invariance of the observable under free dynamics: $  \e^{\i t L_\res}\d\Ga(b_k) =  \d\Ga(b_k(t))= \d\Ga(b_k)  $.  Finally, if $b=b_x$, then the above properties do not hold in general, but we still have invariance under joint translations of $A$ and the final time $n$, i.e.\ 
\beq 
G_{A\cup \{n+1\},n} =  G_{(A+\tau)\cup \{n+\tau+1\},n+\tau}
\eeq
from $ \e^{-\i t L_\res} \initialresfinite=\initialresfinite$.
 
\subsubsection{The contraction operator $\caT$} \label{sec: contraction}
 
We  define the "contraction operator" $\caT: \scrR_A\to \scrR$, by first giving its action on elementary tensors: Consider a family of operators $V_{\tau}  \in \scrR$,  and set 
\beq
\caT \left[  \mathop{\otimes}\limits_{\tau \in A}  \bsI_{\tau}[V_{\tau}]     \right] =V_{\tau_m} V_{\tau_{m-1}}\dots  V_{\tau_{1}}, \qquad  \textrm{where}  \qquad  \tau_m > \tau_{m-1} > \ldots  > \tau_1,  \label{eq: def contraction}
\eeq
and then extend linearly to the whole of $\scrR_A$.  On the left hand side, we will from now on abbreviate $\bsI_\tau [V_\tau] $ by $V_\tau $.  This is a slight abuse of notation that should not cause confusion because we keep the tensor products explicit in the notation, as explained in Section \ref{sec: operator correlation functions}.

By expanding $U_{\tau}=T\otimes \lone+B_\tau$ for every $\tau $ in the expression for the reduced dynamics \eqref{reddyn}, we arrive at
\beq
Q_{n} =       \sum_{A  \subset I_{0,n}  }   \caT   \left[   G_A \mathop{\opprod}\limits_{ \tau \in I_{0,n} \setminus A} T_\tau       \right]  \label{eq: Q from correlation functions} 
\eeq
where, for $A= \emptyset$, we mean to omit  $G_A$ from the right hand side. 
Similarly, for \eqref{reddyn1} we get
\beq
{Q}_{n \str b} =       \sum_{A  \subset I_{0,n+1}  }   \caT   \left[   G_A \mathop{\opprod}\limits_{ \tau \in I_{0,n+1} \setminus A} T_\tau       \right].   \label{eq: Qb from correlation functions} 
\eeq
It is clear that, in the latter formula,  only $A$ with $n+1 \in A$ give a non-zero contribution because $T_{n+1}=0$.

\subsubsection{Connected correlations}

Analogously to classical probability, 
we define the {\it connected correlation functions} or cumulants $G^c_A\in\scrR_A$ for nonempty $A$,  satisfying
$$
G_A=\sum_{\caA} \mathop{\opprod}\limits_{A \in \caA}G^c_A
$$
where $\caA$ run through the set of partitions of $A$.
As in the classical setup, $G^c_A$ can be solved from this inductively in $|A|$, i.e.\
\begin{align}
G^c_{\tau} = G_{\tau}, \qquad    G^c_{\{\tau_1,\tau_2\}} =  G_{\{\tau_1,\tau_2\}}- G^c_{\tau_2}\otimes G^c_{\tau_1},\\[2mm]
G^c_{\{\tau_1,\tau_2,\tau_3\}} =   G_{\{\tau_1,\tau_2,\tau_3\}} -   \sum_{j =1,2,3}  G^c_{\tau_j} \otimes  G^c_{\{\tau_1,\tau_2,\tau_3\} \setminus \{\tau_j\}} - \mathop{\otimes}\limits_{j=1,2,3}  G^c_{\tau_j}. 
\end{align}
Note that it is a consequence and advantage of our conventions that the order in which we write the tensors does not matter.
We obtain then
\begin{align}
Q_{n}  &=       \sum_{\caA \in \frB^0_{0,n} }   \caT   \left[  \mathop{\opprod}\limits_{A \in \caA}  G^c_A   \,  \mathop{\opprod}\limits_{ \tau \in  I_{0,n} \setminus \supp \caA} T_{\tau}    \right],  \label{eqZ from connected correlation functions}  \\
{Q}_{n \str b}  &=       \sum_{\caA \in \frB^0_{0,n+1} }   \caT   \left[  \mathop{\opprod}\limits_{A \in \caA}  G^c_A   \,  \mathop{\opprod}\limits_{ \tau \in  I_{0,n+1} \setminus \supp \caA} T_{\tau}    \right].  \label{eqZ from connected correlation functions boundary} 
\end{align}
 It is immediately clear that any contribution to the sum  in \eqref{eqZ from connected correlation functions boundary}  vanishes unless $n+1 \in \supp \caA$ because of $T_{n+1}=0$.

 \subsubsection{Norms}

Let us introduce a convenient norm on the spaces $\scrR_A$. 
For $E \in \scrR$ (i.e.\ the case $\str A\str=1$), we set
\beq
\norm E \normw  := \norm E \norm=     \sup_{\rho \in \scrB_1({\scrH_\sys}), \norm \rho  \norm_1=1}   \norm  E(\rho)  \norm_1, \eeq
i.e.\ the natural operator norm on $\scrB(\scrB_1(\scrH_\sys))$. 

For $E \in \scrR_A$ with $ 1< \str A \str <\infty$, we exploit that $E$ can be  written  as a finite sum of elementary tensors 
\beq
E = \sum_\nu     E_\nu,   \qquad E_{\nu}=\otimes_{\tau \in A} E_{\nu,\tau}, \qquad  E_{\nu,\tau} \in \scrR_\tau,
\eeq
to  define
\beq \label{def: weird norm}
\norm E \normw:= \inf_{\{ E_\nu \}}  \sum_{\nu} \prod_{\tau \in A} \norm E_{\nu, \tau} \norm  
\eeq
where the infimum ranges over all such elementary tensor-representations of $E$.
This norm  is useful because of  the following properties (trivial from the definition):  
\ben
\item 
For any family of operators  $K_{A \in \caA}$ with  $K_A \in \scrR_A$ and $\caA$ a collection of disjoint sets,  we have
\beq
\left \norm    \mathop{\opprod}\limits_{A \in \caA} K_A   \right \normw  \leq    
\mathop{\prod}\limits_{A \in \caA}  \norm K_A \norm_{\weird}.    
\label{eqbound w norm}
\eeq
\item For any  $K_A \in \scrR_A$, 
\beq
\left \norm   \caT \left[  K_A  \right] \right \norm   \leq    
\left \norm   K_A   \right \normw. 
\eeq
\een

\subsection{Scalar polymer weights}  \label{sec: scalar polymer model}

The representations (\ref{eqZ from connected correlation functions}, \ref{eqZ from connected correlation functions boundary}) 
 evoke the picture of a leading dynamics $T$ interrupted by excitations, indexed by the sets $ A \in \caA$, and with operator valued weights $G^c_A$. We will now construct a similar representation, but with scalar weights.
We exploit the dissipativity of the model, captured in the upcoming lemma. For operators $W
' \in \scrB_1(\scrH_\sys), W
 \in \scrB(\scrH_\sys)$, we write $\str W
' \rangle \langle W
 \str$ to denote the operator in $\scrR$ acting as $S \mapsto  \str W
' \rangle \langle W
\str S =  W
' \Tr( W
^* S)$ 
 \begin{lemma} \label{lem: spectral gap}
Recall the operator $T \in \scrR$ defined in \eqref{def: t}. It has a simple eigenvalue equal to $1$, corresponding to the one-dimensional spectral projector $R=\str \eta \rangle \langle \lone \str$, with $\eta$ a density matrix, such that
 \beq
 \norm T^m - R \norm \leq C \e^{-g m}
 \eeq
 for some $g>0$.
 \end{lemma}
 This is Lemma 2.3 1) in \cite{deroeckkupiainenphotonbound}  specialised to the case $\ka=0$.  We exploit this to split
\beq\label{Tdeco}
T= R   +   T^{\perp} , 
\eeq
where $ T^{\perp}:=T-R$ and we have 
\beq\label{TtauR}
R    T^{\perp} =T^\perp    R =0,  \qquad TR=RT=R.
\eeq
Analogously, we define $T_0=TR+T^\perp_0= R +(\lone-R)$ (since $T_0=\lone$) so that \eqref{TtauR} also holds for $T_0$. 
We will insert these decompositions into the expansions (\ref{eqZ from connected correlation functions}, \ref{eqZ from connected correlation functions boundary}). The following definition provides the tools for this

\begin{definition}[Fusions] \label{def: scalar weights}  Let $\caA \in \frB^0_{0,n}$  and let  $\caJ \in \frB^{1}_{0,n}$ with the property that all $J \in \caJ$  are intervals.
We say that  a pair $( \caA, \caJ) $ is a \emph{fusion}   if
\begin{enumerate}
\item $\supp \caA \cap \supp \caJ =\emptyset$. 
\item   $\distance( I_{0,n} \setminus \supp ( \caA\cup\caJ) , \supp \caJ  )>1$.
\item  The following undirected graph $\Gamma(\caA,\caJ)$ is connected.   Its vertex set is the disjoint  union $\caA \sqcup \caJ$, and its edges are  $\{A,J\}$ with $A \in \caA, J\in \caJ, \distance(A,J)=1$ and $\{A,A'\}$ with $A,A' \in \caA,  \distance(A,A')=1$.
\end{enumerate}
The set of fusions is denoted by $ \frS_n^f $.
\end{definition}
\begin{remark}\label{rem: emtpy fusions}
The only fusions $(\caA,\caJ)$ with $\caA=\emptyset$ are $(\emptyset,\emptyset)$ and $(\emptyset,\{ I_{0,n}  \})$. The  fusion $(\emptyset,\emptyset)$ will not play any role in what follows because its support, $ \supp (\caA\cup\caJ) $, is empty. 
\end{remark}
Define now, for a fusion $( \caA, \caJ)$,
\beq
V((\caA,\caJ)) :=     \mathop{\opprod}\limits_{A \in \caA} G^c_{A} \mathop{\opprod}\limits_{\tau \in\supp \caJ}  T^{\perp}_{\tau}        \label{eqdef polymer weight big v with collections}
 \eeq
as an operator in $\scrR_{\supp( \caA\cup\caJ)}$. By summing fusions with the same support, we set 
\beq
 {\scriptsize \Si}V (A) :=    \sum_{ (\caA,\caJ) \in \frS^f_n: \supp (\caA\cup\caJ) =  A }  V((\caA,\caJ)).        \label{eqdef polymer weight big v}
\eeq
We can now regroup terms in \eqref{eqZ from connected correlation functions} such that \beq
 Q_{n} = \sum_{\caA \in \frB^1_{0,n}}   \caT   \left[     \mathop{\opprod}\limits_{ \tau \in (\supp \caA)^c} R_{\tau} \,  \mathop{\opprod}\limits_{A \in \caA}  
 {\scriptsize \Si}V(A)    \right]   \label{eq: q in terms of primed a}
\eeq
where $(\supp \caA)^c = I_{0,n} \setminus \supp \caA $. We refer the reader to \cite{deroeckkupiainenphotonbound} for a step by step derivation of this formula, that starts by splitting $T_\tau= R_\tau+T^\perp_\tau$ in \eqref{eqZ from connected correlation functions}. 

Note that since $ \caA \in \frB^1_{0,n}$ the sets $A\in\caA$ above are non-adjacent, i.e.\ distances between them are greater than $1$.  Hence, for any $\caA$ in the formula above, all $\tau$ that are adjacent to the set $\supp\caA$ carry the rank-one operator $R$.  
A pictorial way to phrase this is that any of the operators $ {\scriptsize \Si}V (A) $ in \eqref{eq: q in terms of primed a}  is surrounded by projections $R$, except possibly at the boundaries of the interval $I_{0,n}$.  We exploit this by defining, for $A \neq \emptyset$,
\beq
\hat v(A) :=       \caT \left[  {\scriptsize \Si}V(A)  \bigotimes_{\tau \in I_{0,n} \setminus A}  R_{\tau}  \right], \qquad  \hat v(A) \in \scrR.        \label{eqdef polymer weight v} 
\eeq
Note that $\hat v(A)$ is a multiple of $R$ unless $0\in A$ and/or $n\in A$.
Finally, we recall that $R= \str \eta \rangle \langle \lone \str$ and define
\beq 
 v(A)  :=  \left\{ \begin{array}{lr}            \langle \lone,  \hat v(A)  \eta  \rangle & \qquad      \initial \notin A   \\[3mm]
        \langle \lone,  \hat v(A)  \rho_{\sys,0}  \rangle & \qquad      \initial \in A  
      \end{array} \right.    \label{def: boundary values}  
\eeq
With these definitions, one can check that we obtain
\beq
Z_n(\lone,\rho_0)= \Tr Q_{n} \rho_{\sys,0}=   \sum_{\caA \in \frB^1_{0,n}}       \prod_{A \in \caA}  v(A)    \label{eqscalar polymer model boundaries}
\eeq
where we have  used the fact that  $\Tr \rho_{\sys,0}=  \langle \lone, \rho_{\sys,0} \rangle =   1$ to simplify the formula, and the summand on the right hand side is understood to be $1$ for $\caA=\emptyset$. 
Again, a more detailed derivation can be found in \cite{deroeckkupiainenphotonbound}
(compared to the corresponding expression in \cite{deroeckkupiainenphotonbound} the factors $k_\realinitial k_\realfinal$ are missing, $k_\realinitial$ is missing because $\Tr\rho_{\sys,0}=1$ and $k_\realfinal$ is missing because, unlike in \cite{deroeckkupiainenphotonbound}, we don't have an observable consisting of Weyl-operators). 
In the special case where $\rho_0=\eta\otimes P_{\Om}$, using $\Tr \rho_0=1$, \eqref{eqscalar polymer model boundaries} reduces to
\beq \label{eq: z only bulk}
1= Z_n(\lone,\eta\otimes P_{\Om})=  \sum_{\caA \in \frB^1_{1,n}}       \prod_{A \in \caA}  v(A)
\eeq
because in that case $v(A)=0$ whenever $0 \in A$. This follows from $\psi_{\realinitial}=0$ and $T^\perp\eta=0$. 
\begin{remark}\label{rem: fusions with zero weight}
Fusions $(\caA,\caJ)$ with $n \in \supp \caJ$ do not contribute to $v(\cdot)$.  Indeed, they contribute to $\hat v(\cdot)$ an 
%
%
operator of the form $T^\perp K$ for some $K \in \scrR$, but we have
$$ \Tr (T^\perp  K \rho)= \Tr TK\rho -\Tr R K\rho =  0 $$
because $T$ and $R$ conserve the trace. 
In particular, by Remark \ref{rem: emtpy fusions}, fusions with $\caA=\emptyset$ do not contribute. 
\end{remark}
It remains to generalise this formula to the case where we have the observable $\d \Gamma(b)$. As already indicated, this is taken care of by defining the boundary element $n+1$. One could generalise the concepts above, like fusions, to include this element in an appropriate way, but we prefer not to do this, the reason being that the boundary element $n+1$ behaves in a very distinct way.  Instead, we proceed as follows: 
Fix a fusion $(\caA,\caJ)$ with $\caA \neq \emptyset$ and a set  $A \in \caA$.  We modify the collection $\caA$ by replacing the set $A$ by $A \cup \{n+1\}$ and calling the obtained collection $\caA_A$, i.e.\
\beq 
\caA_{A} := (\caA \setminus \{A\} )\cup  \{A \cup \{n+1\} \}.
\eeq
We can then define the operator  $V((\caA_A,\caJ))$ via \eqref{eqdef polymer weight big v with collections} as an operator on $\scrR_{\supp (\caA\cup \caJ) \cup \{n+1 \}}$ because   $G^c_A$ with $ n+1
\in A$ is well-defined. Then  we set
\beq
 {\scriptsize \Si}V (A' \cup \{n+1\}) :=   \sum_{ \substack{ (\caA,\caJ) \in \frS^f_n,  \caA \neq \emptyset \\  \supp (\caA\cup\caJ) =  A' } }  \,   \sum_{A \in \caA} V((\caA_{A} ,\caJ)),        \label{eqdef polymer weight big v boundary}
\eeq
and we simply define $\hat v(A \cup \{n+1\})$ and $v(A \cup \{n+1\})$ by the relations   \eqref{eqdef polymer weight v} and  \eqref{def: boundary values} with $A$ replaced by $A \cup \{n+1\}$.  For consistency with later formulas, we also set $v(\{ n+1\})=0$. 
 Note that we do not extend the setup to include the possibility that $n+1 \in \supp\caJ$. This is indeed not necessary since such a contribution would necessarily vanish because $T_{n+1}=0$, see \eqref{Tfinal}. 
Now, the final expression for $Z_n(\d\Ga(b), \rho_0)$ reads
\beq \label{eq: basic rep zn}
Z_n(\d\Ga(b), \rho_0)=     \sum_{\caA \in \frB^1_{0,n}}     \sum_{A \in \caA}   v(A \cup \{ n+1\})   \prod_{A' \in \caA \setminus \{A\} }  v(A')  
\eeq
where it is understood that $\caA = \emptyset$ does not contribute to the right hand side and the empty product is set to $1$.

\subsection{Estimates on operator-valued polymers} \label{sec: estimates}

\subsubsection{Dyson expansion}
We will now derive a formula for the correlation functions $G_A^c$ in graphical terms.
Recalling that $H=H_S+H_F+H_I$ we decompose $L=\ad(H)$ as
\beq
L=L_\res+ L_\sys+ L_\inter
\eeq
and  introduce
 \beq 
 L_\inter(s)=  \begin{cases} \e^{\i s L_\res}L_\inter  \e^{-\i s L_\res} &  s\geq 0\\   \ad(\Phi(\psi_\realinitial)) & -1 \leq s <0 \end{cases}
 \eeq
We develop the evolution operator $\e^{-\i t L}$ and the Weyl operator $\caW(\psi_\realinitial)$ in a standard way in a Dyson expansion, arriving at
\baq \label{eq: first duhamel series hat}
\e^{\i t L_\sys} Q_{n} \rho_{\sys,0}=    \sum_{m \in \bbN}  (-1
)^{m}\mathop{\int}\limits_{-1\leq t_1 < \ldots < t_{2m} <n/\la^2} \d t_1 \ldots \d t_{2m} \,   \Tr_{\res}\left[  L_{\inter}(t_{2m})  \ldots  L_{\inter}({t_2})  L_{\inter}({t_1})  (\rho_{\sys,0}\otimes P_{\Om})\right]. 
\eaq
Since the operators $L_\inter$ are unbounded, the formula and its derivation require justification that we provide in \cite{deroeckkupiainenphotonbound}.
The
 integrand can be written in terms of the formalism developed in Section \ref{sec: operator correlation functions}  with obvious modifications\beq
 \Tr_{\res}\left[  L_{\inter}(t_{2m})  \ldots  L_{\inter}({t_2})  L_{\inter}({t_1})  (\rho_{\sys,0}\otimes P_{\Om})\right]
  = \left( \caT  \bbE\left[L_{\inter}(t_{2m})  \otimes_\sys\ldots  \otimes_\sys L_{\inter}({t_2}) \otimes_\sys  L_{\inter}({t_1})  \right] \right)\rho_{\sys,0}.
\eeq
These modifications will not be discussed here in detail (see \cite{deroeckkupiainenphotonbound}). Briefly said, 
we introduce copies of $\scrR$ indexed by the times $t_1,t_2,\ldots,t_m$ and labelled products of them. For example, the term $ \bbE\left[ \ldots \right]$ above is an element of $\scrR^{{\otimes^m}}$ that we identify with an element of $ \scrR_{\{t_1,\ldots,t_m\}}$, and the operator $\caT$ contracts it into an element of $\scrR$. 
Applying Wick's theorem, one gets
\beq \label{eq: wick gives pairs}
 (-1
)^{m}
 \bbE\left[L_{\inter}(t_{2m})  \otimes_\sys\ldots  \otimes_\sys L_{\inter}({t_2}) \otimes_\sys  L_{\inter}({t_1})  \right] =\sum_{\pi \in \textrm{Pair}(t_1,\ldots,t_{2m})}  \mathop{\otimes}\limits_{(u,v ) \in \pi}  K_{u,v}  
\eeq
where  $K_{u,v} $ is defined as an operator in $\scrR^{{\otimes^2}}$:
\beq
K_{u,v} = -\bbE( L_{\inter}(v)  \otimes_\sys  L_{\inter}(u)  ),
\eeq
identified with an element of $\scrR_{\{u,v\}}$, and 
 $ \textrm{Pair}(t_1,\ldots,t_{2m})$ denotes the set of pairings of the set $\{t_1,\ldots,t_{2m}\}$, and we write  the pairs as ordered pairs $( u,v)$ with the convention $u\leq v$. 
Substituting \eqref{eq: wick gives pairs} in \eqref{eq: first duhamel series  hat}  we arrive at
\beq \label{eq: 2duhamel series }
\e^{\i t L_\sys} Q_{n}=     \sum_{m \in \bbN}  \mathop{\int}\limits_{-1\leq t_1 < \ldots < t_{2m} <n/\la^2} \d t_1 \ldots \d t_{2m}
\sum_{\pi \in \textrm{Pair}(t_1,\ldots,t_{2m})}  \, \,  \caT\left[  \mathop{\opprod}\limits_{w \in \pi } K_{w}    \right]
\eeq
where we abbreviate the pairs as $w=(u,v)$. 
In \cite{deroeckkupiainenphotonbound} it is explained how this expression may be written
as an integral in a suitable space. Consider a set whose elements are families $\uw$
of pairs of times:  $\uw=\{w_1,w_2,\dots,w_m\}$ with $m\geq 0$
and $w_i=(u_i,v_i )$, $u_i\leq v_i$ and  $u_i,v_i\in [-1,n/\la^2]$. This set carries a  $\sigma$-algebra
and a measure $\mu(\d \uw)$ so that \eqref{eq: 2duhamel series } becomes
\beq 
\e^{\i t L_\sys} Q_{n}=
\mathop{\int}
 \mu(\d \uw)   \caT\left[  \mathop{\opprod}\limits_{i } K_{w_i}    \right].
 \label{eq: 4th duhamel series}
\eeq
It is understood that $\uw=\emptyset$ contributes $1$ to the right hand side. Let us now additionally define
\beq
K_{w \str b} := -\bbE( U_{n+1}\otimes_\sys L_{\inter}(v)  \otimes_\sys  L_{\inter}(u)  ),\qquad w=(u,v),\eeq
as an operator in $\scrR^{{\otimes^3}}$ that we identify with $\scrR_{\{u,v, t \}}$ (recall  $t=n/\la^2$)  such that the operator $U_{n+1}$ acts on the space indexed by $t$. 
Then, the expansion \eqref{eq: 2duhamel series } can also be performed in the presence of the observable $\d\Gamma(b)$:
\begin{align}
\e^{\i t L_\sys} {Q}_{n \str b} &  =     \sum_{m \in \bbN}  (-1
)^{m}\mathop{\int}\limits_{-1\leq t_1 < \ldots < t_{2m} <n/\la^2} \d t_1 \ldots \d t_{2m} \, \,  \Tr_{\res}\left[ U_{n+1} L_{\inter}(t_{2m})  \ldots  L_{\inter}({t_1})  (\cdot\otimes P_{\Om})\right] \nonumber\\[2mm]
&=   \sum_{m \in \bbN}  \mathop{\int}\limits_{-1\leq t_1 < \dots < t_{2m} \leq n/\la^2} \d t_1 \ldots \d t_{2m} \, \,  \sum_{\pi \in \textrm{Pair}(t_1,\ldots,t_{2m})}   \sum_{w_0 \in \pi} \caT\left[  K_{w_0 \str b} \mathop{\opprod}\limits_{\substack{w \in \pi \\  w \neq w_0 }}  K_{w}  \right]   \nonumber\\[2mm]
&=
\mathop{\int}
 \mu(\d \uw)  \sum_{i}  \caT\left[ K_{w_i \str b} \mathop{\opprod}\limits_{j\neq i} K_{w_j}    \right].
 \label{eq: 5th duhamel series}
\end{align}

We proceed with the identification of $G^c_{A}$ from these expansions. To do that
we need to coarse grain them to the  macroscopic time scale (in units of $1/\la^2$).
Given an $s\in [-1,n/\la^2]$ let $[s]$ denote the smallest integer not smaller than $\la^2s$
i.e. $s\in   ]  \la^{-2} ([s]-1),   \la^{-2}[s]  ] $.
Then, 
given  $\uw=\{w_1,w_2,\dots,w_m\}$ let $[\uw]\subset \bbN$ be the union of
the  $[u_i]$ and  $[v_i]$ for $w_i=(u_i,v_i)$. 

The contraction operator $\caT[\cdot]$ defined in Section \ref{sec: contraction} contracts operators from $\scrR_A$ to $\scrR$. We now define a contraction operator $\caT_A$ that produces operators in $\scrR_A$.
Let us consider a finite family of operators $V_{t_i} \in \scrR_{t_i} $ where the indexed times $t_i$ satisfy $t_i < t_{i+1}$ and $[t_i]  \in A$. Then we set
\beq
\caT_{A} \left[  \mathop{\otimes}\limits_{i }  V_{t_i}   \right]   :=         \mathop{\otimes}\limits_{\tau \in A}  \bsI_{\tau} \left[  \caT \left[
\mathop{\otimes}\limits_{j:  [t_j] =\tau } V_{t_j} \right] \right]
\eeq 
and we extend by linearity to the whole of $\otimes_i\scrR_{ t_i }$, obtaining   $\caT_{A}: \otimes_i\scrR_{ t_i } \mapsto \scrR_A$. In words, $\caT_{A}$ puts each operator into the right 'macroscopic' time-copy and contracts the operators within each macroscopic time-copy. 
Coarsegraining \eqref{eq: 4th duhamel series} this way leads to the formula
\beq  
\tilde Y_{A}    G_A Y_{A} =    \mathop{\int}\mu(\d \uw)  \indicator_{[\uw]=A}  \caT_A\left[  \mathop{\opprod}\limits_{i} K_{w_i}    \right].
\eeq
The factors $\tilde Y_{A} $ and $  Y_{A}$ come from the  free $\sys$-evolutions  in 
\eqref{eq: 4th duhamel series}  and the definition of $ G_A$. They are defined as 
\beq
 Y_A =  \mathop{\otimes}\limits_{\tau \in A \setminus \{0\}}Y_\tau,  \qquad     \widetilde Y_A =  \mathop{\otimes}\limits_{\tau \in A \setminus \{0\}}  \widetilde Y_\tau,
\label{eq: y operators}
\eeq 
with \beq
Y_\tau =   \bsI_{\tau}[\e^{\i (\tau-1) L_\sys}], \qquad   \widetilde Y_\tau =    \bsI_{\tau}[\e^{-\i \tau L_\sys}].\
\eeq 
Since $\e^{-\i \tau L_\sys}$ is an isometry in the operator norm of $\scrB_1(\scrH_\sys)$, left and right multiplication by $Y_A,\widetilde Y_A$ is an isometry on $\scrR_A$ in the norm $\norm \cdot \normw$, and therefore 
$\widetilde Y_{A} $ and $ Y_{A}$ play no
role in what follows.

The connected correlations  $G^c_A$ have  similar quite obvious expressions. 
Given a   $\uw$ we can define an undirected graph $\caG(\uw)$ with vertex set $[\uw]$ and
edges  $\{\tau,\tau'\},  \tau \leq\tau'$ whenever  there is a pair $w_i=(u_i,v_i)$ such that  $[u_i]= \tau$
and  $ [v_i]= \tau'$.   Let moreover
\begin{align} \caC(A) & := \{ \uw \,  \big\str \,   [\uw]=A\, \,  \text{and $\caG(\uw)$  is  connected} \}, \\[2mm] 
\caC'(A) & := \begin{cases} \{ \uw \,  \big\str \,   [\uw]=A \}  &  \str A \str =1   \\[1mm]     \caC(A)  &  \str A \str >1.    \end{cases} \end{align}
We have then
\begin{lemma}
\label{lem: identification correlation diagrams}  Let $A\in I_{0,n}$. Then 
\beq    \label{eq: b connected in diagrams identity}
 G^c_A  \cong    \mathop{\int} \mu(\d \uw)\indicator_{\caC(A)}  \caT_A\left[  \mathop{\opprod}\limits_{i} K_{w_i}    \right],
\eeq
\beq     \label{eq: b connected in diagrams identity final}
  G^c_{A \cup \{n+1\}} \cong  \mathop{\int} \mu(\d \uw)\indicator_{\caC'(A)} \sum_{i} \caT_A\left[  K_{w_i \str b} \mathop{\opprod}\limits_{j\neq i} K_{w_j}    \right]
\eeq 
 where  $\cong$ denotes an  isometry in the norm $\norm\cdot \normw$.  
\end{lemma}
The obvious proof of \eqref{eq: b connected in diagrams identity} is in  \cite{deroeckkupiainenphotonbound}, the proof of \eqref{eq: b connected in diagrams identity final} is analogous.   By Lemma \ref{lem: identification correlation diagrams}, and the properties of the norm $\norm \cdot \norm_\diamond$, we immediately get the bounds
\beq    \label{eq: b connected in diagrams}
\norm   G^c_A \normw \leq   \mathop{\int} \mu(\d \uw) 
\indicator_{\caC(A)} 
 \prod_{i=1}^m  \norm K_{w_i}   \normw,
\eeq
\beq    \label{eq: b connected in diagrams boundary}
\norm   G^c_{A \cup \{n+1\}} \normw \leq   \mathop{\int} \mu(\d \uw) \indicator_{\widetilde\caC(A)}  \sum_{i=1}^m  \norm  K_{w_i \str b}  \norm_{\diamond}   \prod_{j\neq i}  \norm   K_{w_j}   \normw.
\eeq

\subsubsection{Bounds on the operators $K_{w}, K_{w \str b}$ and  $G^c_A$ }

To bound the operators $K_{w}$, we first have to address the fact that these operators are qualitatively different whenever one or both of the times $\{u,v\}$ is smaller than $0$ (because then it originates from the expansion of the Weyl operator, rather than from the interaction).  Let us write (recall the form factor $\phi$)
\beq
\phi_s =  1_{s\geq 0} \e^{\i s \om } \phi+   1_{s< 0}  \psi_\realinitial.
\eeq
Then we define the functions $h(u,v)$ and $h(u,v \str b)
$ by 
\beq
 \str \la \str^{1_{u\geq 0} + 1_{v\geq 0}}h(u,v) := \norm K_{u,v} \norm_{\diamond},\qquad    \str \la \str^{1_{u\geq 0} + 1_{v\geq 0}}  h(u,v \str b)
 := \norm K_{u,v \str b} \norm_{\diamond} 
\eeq
where we should however  keep in mind that  $ h(u,v \str b_k)
$ depends on $\delta$ and    $ h(u,v \str b_x)
$  depends on $t_c$ and the final time $t$. We usually do not indicate this dependence (see however item 1) of  Proposition \ref{prop: bounds on correlation functions}).
From the definition of the norm $\norm \cdot \norm_\diamond$ and the definition of $U_{n+1}$ in \eqref{def: u final} we have
\beq
h(u,v) \leq    4 \norm D \norm^2 \str \langle \phi_{v}, \phi_{u} \rangle_{\frh} \str, \qquad  h(u,v \str b)
    \leq   4 \norm D \norm^2   \str  \langle  \phi_v,  b(t) \phi_u \rangle_\frh \str
 \eeq
where $b(t)= \e^{\i \om t} b \e^{-\i \om t}$.
The important properties of the  functions $h(u,v), h(u,v \str b)
$ are  collected in
 \begin{proposition}[Bounds on correlation functions] \label{prop: bounds on correlation functions}
 Unless mentioned otherwise, let $u>-1$. 
 \begin{enumerate}
\item  If $u\geq 0$ and $s\geq 0$, then 
 \beq h(u,v)= h(u+s,v+s), \qquad h(u,v \str b_k)
= h(u+s,v+s \str b_k). \eeq 
The function $h(u,v \str b_x)
$ depends on the final time $t$ and in general $h(u,v \str b_x)
 \neq h(u+s,v+s \str b_x)$. We can indicate this dependence by writing $h(u,v, t \str b_x)$, then 
\beq 
h(u,v, t \str b_x)=  h(u+s,v+s, t+s \str b_x).
\eeq
\item 
 \beq  \int_u^{t} \d v   \langle v-u \rangle^{1+\al}  h(u,v)    \leq C 1_{u\geq 0} + \breve C 1_{u < 0}.    \label{eq: standard bound h} \eeq    
\item 
   \beq  \int_u^{t} \d v   \langle v-u\rangle^{1+\al/2}  h(u,v \str b_k)
    \leq \breve C \delta^{\al/2}. \eeq   
\item   
     \beq 
  \int_u^{t} \d v   \langle v-u\rangle^{1+\al}  h(u,v \str b_x)_{\sup_t}
    \leq \breve C.   \label{eq: nonstandard bound h} \eeq 
    where $h(u,v \str b_x)_{\sup_t}:= \sup_{q \in \bbR}  \big(h(u+q,v+q \str b_x) 1_{v+q \leq t} 1_{u+q \geq  -1}\big)  $. 
\item  Recall $r_\theta <1$ is the radius of a ball containing $\supp \theta$. Fix a number $m_\theta$ such that  $r_{\theta} <  m_\theta < 1$, then
   \beq   \int_{-1}^{t-m_\theta t_c}  \d u\int_u^{t} \d v   \,   \langle t-m_\theta t_c-u \rangle^{\al}   h(u,v \str b_x)
    \leq \breve C \eeq   
    where $\breve C$ can depend on $m_\theta$. 
 \end{enumerate}
Whenever applicable, the bounds above are uniform in $u$ and $t$. 
 \end{proposition}
 Note that in item 4), $h(u,v \str b_x)_{\sup_t}$ differs from $h(u,v \str b_x)$ in that; unlike the latter, it is a function of $v-u$, i.e.\ it is translation invariant. The same remark applies to the upcoming bounds \eqref{eq: sum edge factors b x} and \eqref{eqbound scalar polymers mpos general}. 
 Item $1)$ follows immediately from the fact that $\theta(k/\delta)$ commutes with $\e^{\i s \om}, s \in \bbR$ and the group property $ \e^{\i s \om}\e^{\i s' \om}= \e^{\i (s+s') \om}$. 
The proofs of the other claims concern only the one-boson problem and  they are of a completely different nature than the rest of this paper. Therefore, we gather those proofs in  Appendix \ref{app: propagation estimates}.

\subsubsection{Bounds for operator-valued polymers }

The next step is to use the bounds on the $h$-functions and the formulae (\ref{eq: b connected in diagrams},\ref{eq: b connected in diagrams boundary}) to derive bounds on the operator-valued correlation functions $G_A^c$.

Consider first \eqref{eq: b connected in diagrams}. Let $\scrT(A)
$ be the
set of  $\uw$ such that $ [\uw]=A$ and $\caG(\uw)$  is  a (connected) tree. Then
\beq    \label{eq: b connected in diagrams1}
   \mathop{\int} \mu(\d \uw) 
\indicator_{\caC(A)} 
 \prod_{i=1}^m  \norm K_{w_i}   \norm_{\diamond}\leq  \mathop{\int} \mu(\d \uw') 
\indicator_{\scrT(A)
} 
 \prod_{i=1}^{m'}  \norm K_{w'_i}  \norm_{\diamond}  \mathop{\int} \mu(\d \uw'') 
\indicator_{[\uw'']\subset A} 
 \prod_{i=1}^{m''}  \norm K_{w''_i}   \norm_{\diamond}.
\eeq
Indeed, the pairings and integrals on the left hand side form a subset of the ones on the
right hand side: since  $\caG(\uw)$ is  connected it contains a (in general not unique) spanning tree $\scrT$ (a tree with the same vertex set as the total graph, i.e.\  $\caG(\uw)$)
and thus there is a subset   $\uw'$  of $\uw$ so that $\caG(\uw')=\scrT$. The remaining set of pairs
$\uw''$ in $\uw$ meets the constraint ${[\uw'']\subset A} $.

We first perform the integral over $\uw''$. The integrability results in Proposition \ref{prop: bounds on correlation functions} lead to the estimate
\beq
 \mathop{\int} \mu(\d \uw'') 
\indicator_{[\uw'']\subset A} 
 \prod_{i=1}^{m''}  \norm K_{w''_i}   \norm_{\diamond}\leq 
(1+C_{\realinitial}1_{0 \in A })\e^{ C\str A  \str }.
\eeq
  This is explained in detail in \cite{deroeckkupiainenphotonbound} (see the proofs of Lemma 3.1 and Lemma 3.4 therein).  
To perform the integral over  $\uw'$ let us define for $\tau,\tau'\in \bbN$, $\tau<\tau'$
  \beq  \label{def: e factor}
\hat{e}(\tau, \tau')=  \str \la \str^{1+1_{\tau>0}}\int_{\Dom(\tau)}    \d u \int_{\Dom(\tau')}   \d v \,   h(u,v).
\eeq 
  Here we use the notation  $\Dom(\tau)=]\la^{-2}(\tau-1),\la^{-2}\tau]$ for $\tau>0$ and $\Dom(0)=]-1,0]$, i.e.\ $\Dom (\tau) = \{s \geq -1, [s] =\tau\}$. 
  Then
\beq    \label{eq: b connected in diagrams2}
   \mathop{\int} \mu(\d \uw') 
\indicator_{\scrT(A)
} 
 \prod_{i=1}^{m'}  \norm K_{w'_i}  \norm_{\diamond}\leq 
  \sum_{ \scrT:  \caV(\scrT)=A}
 \prod_{\{\tau,\tau' \} \in\caE(\scrT)}  \hat e (\tau,\tau')   
\eeq
where the sum runs over all trees $\scrT$ whose vertex set $\caV(\scrT)$ is $A$, i.e.\ over all spanning trees on $A$, and $\caE(\scrT)$ is the edge set of the tree $\scrT$. 
Altogether we have obtained
\begin{lemma} Let $A \subset I_{0,n}$, then
\beq \label{eq: g from trees standard}
\norm G^c_{A} \normw \leq  (1 + 1_{0 \in A} \breve C)   \e^{ C \str A\str}  \sum_{ \scrT:  \caV(\scrT)=A}
 \prod_{\{\tau,\tau' \} \in\caE(\scrT)}  \hat e (\tau,\tau').   
\eeq 
\end{lemma}
Let us now derive an analogous bound for $\norm G^c_{A \cup \{n+1\}} \normw $.  We first define, for $\tau \leq \tau'$ (contrary to the above we will need the case $\tau=\tau'$); 
\beq  \label{def: e factor final}
{\hat{e}}(\tau, \tau' \str b) : =   \str \la \str^{1_{\tau>0}+1_{\tau'>0}}  \int_{\Dom(\tau)}    \d u \int_{\Dom(\tau')} \d v \,   h(u,v \str b)
  1_{v\geq u}.
\eeq
In  \eqref{eq: b connected in diagrams boundary}, we distinguish the cases where  $[w_i] = \{\tau,\tau'\}, \tau \neq \tau'$ and $[w_i] =\tau_0$.  In the first case, we make the edge 
$\{\tau,\tau'\}$ part of the spanning tree, in the second case we add the factor $ \hat e(\tau_0,\tau_0 \str b)$ by hand to the product of  edge factors of the spanning tree.   The resulting estimate is
\begin{lemma}  Let $A \subset I_{0,n}$ with $\str A \str >1$, then
\begin{align} 
\norm G^c_{A \cup \{n+1\}} \normw   \leq   \breve C   \e^{ C \str A \str}    \sum_{ \scrT:  \caV(\scrT)=A}
 \Big( \sum_{ \{\tau_0,\tau'_0\}  \in \caE(\scrT)  }  \hat e(\tau_0,\tau'_0 \str b) \prod_{\substack{\{\tau,\tau'\} \in\caE(\scrT) \\ \{\tau,\tau'\} \neq \{\tau^{}_0,\tau_0'\}   }} \hat e (\tau,\tau')  \qquad &  \nonumber \\[1mm]
+\,\,  \sum_{\tau_0 \in A  }  \hat e(\tau_0,\tau_0 \str b) \prod_{\{\tau,\tau' \} \in\caE(\scrT)} \hat e (\tau,\tau') \Big). &   \label{eq: g from trees nonstandard}
\end{align} 
In case $A= \{\tau\} $, we have  simply  $\norm G^c_{ \{ \tau, n+1\}} \normw   \leq   \breve C  \hat e(\tau,\tau \str b)   $. 
\end{lemma}
To proceed, we need bounds on the $\hat e$ factors.  They follow rather straightforwardly from the bounds on $h(u,v), h(u,v \str b)
$.  For convenience we set $\hat{e}(\tau', \tau):= \hat{e}(\tau, \tau')$ and $\hat{e}(\tau', \tau \str b ):= \hat{e}(\tau, \tau' \str b )$. 
For $\hat{e}(\cdot, \cdot)$, we repeat the bound from \cite{deroeckkupiainenphotonbound}. \beq  \label{eq: sum edge factors}
\sum_{\tau' \in I_{1,n} \setminus \{\tau\}}   \langle \tau'-\tau \rangle^{1+\al} \hat{e}(\tau, \tau')  \leq \begin{cases} C \la^2  & \tau \neq 0 \\  \breve C \str\la\str  & \tau = 0. \end{cases}
 \eeq
 To obtain this bound, we bound the sum by (a constant times) the integrals $\int \d u \int \d v$. For $\str \tau'-\tau\str>1$, we gain a factor $\str\la\str^{2(1+\al)}$ by using $ \langle \tau'-\tau \rangle^{1+\al}  \leq \str \la\str^{ 2(1+\al)}  \langle v-u \rangle^{1+\al} $ and item 2) of Proposition \ref{prop: bounds on correlation functions}.  This factor compensates the $\la^{-2}$ coming from the integration over $u$ (in case $\tau>0$) so that the explicit $\str\la\str^{1+ 1_{\tau>0}}$ factor from \eqref{def: e factor} is retained on the right hand side of \eqref{eq: sum edge factors}.   For $\tau'=\tau+1$, we estimate (for definiteness, take $\tau>0$, the other case is trivial)
 \beq
 \int_{(\tau-1)/\la^2}^{\tau/\la^2} \d u   \int_{\tau/\la^2}^{(\tau+1)/\la^2} \d v  \,  h(u,v) \leq    \int_0^{\infty} \d v \str v\str \, h(0,v) \leq C.
 \eeq
 where we used translation invariance  (Item 1) of Proposition \ref{prop: bounds on correlation functions}).
 For $b=b_k$,    we get
\beq  \label{eq: sum edge factors bmom}
\sum_{\tau' \in I_{0,n}}    \langle \tau'-\tau \rangle^{1+\al/2} \hat{e}(\tau, \tau' \str b_k)
  \leq \breve  C\delta^{\al/2}.
 \eeq
Compared to \eqref{eq: sum edge factors}, the term $\tau=\tau'$ is now included in the sum. The derivation proceeds as above, but now starting from items 1,3) of Proposition \ref{prop: bounds on correlation functions}. In case $\tau'=\tau$, one cannot extract any $\la$-dependent small factor from the change of variables $(\tau,\tau')\to (u,v)$ so that the explicit $\str\la\str^{1_{\tau'>0}+ 1_{\tau>0}}$ factor from \eqref{def: e factor final} is used to cancel the $u,v$  integration and therefore there is no $\la$-dependent small factor on the right hand side of \eqref{eq: sum edge factors bmom}.

For $b=b_x$, we similarly derive the analogue of \eqref{eq: sum edge factors bmom}, using Proposition \ref{prop: bounds on correlation functions},
 item 4);
\beq  \label{eq: sum edge factors b x}
\sum_{\tau' \in  I_{0,n}}     \langle \tau'-\tau \rangle^{1+\al} \hat{e}(\tau, \tau' \str b_x)_{\sup_n}  \leq  \breve C  
 \eeq
 where $\hat{e}(\tau, \tau' \str b_x)_{\sup_n} := \sup_{\tau'' \in \bbZ} \big(  \hat{e}(\tau+\tau'', \tau' +\tau''\str b_x) 1_{\tau+\tau'' \geq 0} 1_{\tau'+\tau'' \leq n}\big)   $. Using
 Proposition \ref{prop: bounds on correlation functions}, item 5), we also get
\beq \label{eq: sum edge factors b x special}
\sum_{\tau \leq n-m_\theta n_c } \sum_{ \tau \leq \tau' \leq n}        \langle n-m_\theta n_c-\tau \rangle^{\al}      \hat{e}(\tau, \tau' \str b_x)
  \leq  \breve C
\eeq
where $n_c$ was defined at the beginning of Section \ref{sec: polymer rep} (recall $t_c=\la^{-2}n_c$).

\subsection{Properties of scalar polymers}  \label{sec: bounds on operator valued polymers}

The scalar polymer weights $v(A)$ were defined in Section \ref{sec: scalar polymer model}.  We state some bounds. 
\begin{lemma}[Bounds on scalar polymers]       \label{lem: bound on scalar polymers}   All estimates hold uniformly in $\tau$.
\ben 
\item
For bulk polymers, i.e.\ $0 \not \in A$, we have for  $\tau\in I_{1,n}$, 
\beq
 \sum_{A \subset I_{1,n}:  \,  \tau\in  A  }  \e^{c\str A \str }    \dist(A)^{1+\al}  \str v(A) \str \leq   C \la^2.  \label{eqbound scalar polymers}
\eeq
\item
For polymers containing $0$, we have 
 \beq
 \sum_{A \subset I_{0,n}:  \,  0\in A }  \e^{c\str A \str }   \dist(A)^{1+\al}  \str v(A) \str \leq   \breve C \str \la\str.  \label{eqbound scalar polymers initial}
\eeq
\item  Let $b=b_k$, then  for $\tau\in I_{0,n}$,
\beq 
 \sum_{A \subset I_{n,0}:   \,  \tau \in A }  \e^{c\str A \str }  \dist(A)^{1+\al/2}  \str v(A \cup \{\final\}) \str \leq  \breve C \delta^{\al/2}.  \label{eqbound scalar polymers bmom}
\eeq
\item  Let $b=b_x$, then for  $\tau\in I_{0,n}$,
\beq
 \sum_{A \subset I_{0,n}:  \, \tau \in A }  \e^{c\str A \str }    \dist(A)^{1+\al}  \str v(A \cup \{\final\}) \str_{\sup_n}  \leq   \breve C  \label{eqbound scalar polymers mpos general}
\eeq
where  $\str v(A \cup \{\final\}) \str_{\sup_n} := \sup_{\tau' \in \bbZ} \big(\str v((A+\tau') \cup \{\final\})\str\,  1_{\min (A+\tau') \geq 0} 1_{\max (A+\tau') \leq n}  \big) $.
\item Let $b=b_x$, then
\beq   
 \sum_{A \subset I_{0,n}: \min A \leq  n-m_\theta {n_c} }     \e^{c\str A \str }    \langle n-m_\theta n_c - \min A\rangle^{\al}  \str  v(A \cup \{\final\} ) \str \leq  \breve C.
  \label{eqbound scalar polymers mpos special}
\eeq

\een

\end{lemma}

\subsubsection{Proof of Lemma \ref{lem: bound on scalar polymers}}

First, we restrict to $A \subset I_{1,n}$.  By using the definitions  (\ref{def: boundary values}, \ref{eqdef polymer weight v}, \ref {eqdef polymer weight big v}, \ref{eqdef polymer weight big v with collections}), we can bound the polymer weight $v(A)$ by a product of  $\norm  \cdot \norm_\diamond$-norms of  operators $G_A^c$,  projections  $R$ and $(T^\perp)^{\str J \str} $,  i.e.\
\beq
\str v(A) \str \leq  \sum_{(\caA,\caJ) \in \frS^f_n: \supp (\caA\cup\caJ)=A}  \norm R \norm^{\str \supp \caA \str} \prod_{A' \in \caA}    \norm  G_{A'}^c \norm_\diamond \prod_{J \in \caJ}    \norm (T^\perp)^{\str J \str}  \norm.
\eeq
Next, we use  $\norm R \norm \leq C, \norm (T^\perp)^m \norm \leq C \e^{-m g}$ and the bounds  \eqref{eq: g from trees standard} on $\norm  G_{A'}^c \norm_\diamond$ to get
\beq \label{eq: bound on va norms}
\str v(A) \str \leq   \sum_{(\caA, \caJ) \in \frS^f_n: \supp (\caA \cup \caJ)=A}  \prod_{J \in \caJ} (C \e^{- \str J \str g})  \prod_{A' \in \caA} \, \e^{C \str A'\str}  \sum_{\scrT: \caV(\scrT)=A'}  \,   \prod_{\{\tau,\tau'\} \in \caE(\scrT)}\hat e(\tau,\tau').  
\eeq
Let us now take $ A \subset I_{0,n}$, i.e.\ we allow $0 \in A$, than the bound \eqref{eq: bound on va norms} remains valid
 if we multiply the right hand side by $1+ 1_{0 \in A} \breve C$.  Indeed, the only changes are 1)  at most one of the factors $\norm T^{\perp} \norm $ is replaced by $\norm (T^{\perp})_0 \norm  = C\norm T^{\perp} \norm $ and 2) the bound   \eqref{eq: g from trees standard} has the factor $1+ 1_{0 \in A'} \breve C$ for at most one of the sets $A'$.

We estimate \eqref{eq: bound on va norms}  by viewing the sums on the right hand side as a sum over certain connected graphs.  Let
$\scrS = \caJ\sqcup \caE(\scrT)$, i.e.\ we label the element of $\scrS$ as intervals $(J)$ or edges ($E$). The elements of $\scrS$ are denoted by $S,S'$ and collections of them are denoted by $\caS$.  We write $\supp S$ to denote the subset of $\bbN$ defined by $S$, i.e.\ $S$ without the interval/edge label, and $\supp \caS=\cup_{S\in \caS}\supp S$. 
We assign to any $S \in \scrS$ a weight $w^{(\be)}_{\mathrm{s}}(S)$, with $\be>0$, as follows:
\beq
w^{(\be)}_{\mathrm{s}}
(S) :=  c(w)\times \begin{cases}  \breve c(w)  \langle \tau \rangle^{\be}  \str\la\str^{-1} \hat e(0,\tau)  &  S\, \,   \textrm{is the edge} \, E= \{0,\tau \}, \qquad 0 <\tau    \\[2mm]    
  \langle \tau'-\tau \rangle^{\be}  \str\la\str^{-2} \hat e(\tau,\tau')  &  S\, \,   \textrm{is the edge} \, E= \{\tau,\tau' \}, \qquad 0 <\tau < \tau'   \\[2mm]    
   \str J\str^{\beta} \e^{-(g/2) \str J \str }&S\, \,   \textrm{is the interval} \, J
 \end{cases}
\eeq
where $g$ is as in Lemma \ref{lem: spectral gap} and the constants $c(w),\breve c(w)$ will be fixed below. 
We define an adjacency relation $\sims$ on $\scrS$ by 
\begin{align}
J \sims E &   \Leftrightarrow      \distance( J, E ) = 1,   \nonumber    \\ 
E  \sims E' &   \Leftrightarrow     E  \cap E' \neq \emptyset,    \nonumber    \\ 
J  \sims J' &   \Leftrightarrow   J=J'.   \label{def: sims}
\end{align}
Then, using \eqref{eq: sum edge factors} and $g>0$, we can choose $c(w),\breve c(w)$ small enough such that, for any  $\be \leq 1 + \al$
\beq \label{eq: kotecky preiss}
\sum_{S \in \scrS:\,  S \sims S'}    w^{(\be)}_{\mathrm{s}}
(S)  \leq  1/\e,
\eeq
uniformly for small enough $\la$.  We now  claim that, for sufficiently small $c>0$, 
\beq \label{eq: bound by connected sets}
 \e^{c \str A \str} \dist(A)^{1+\al} \str v(A) \str \leq  (\la^2C 1_{0 \notin A} +\str\la\str \breve C 1_{0 \in A} )   \sum_{\substack{\caS \subset \scrS: \,  \supp \caS=A  \\[1mm]   \caS \,  \textrm{connected}   }}     \prod_{S \in \caS}   w^{(1+\al)}_{\mathrm{s}}
(S) 
\eeq
where $\caS \,  \textrm{connected}$ means that the graph with vertex set $\caS$ and edges $\{S,S'\}$ if $S \sims S'$, is connected.
To check \eqref{eq: bound by connected sets}, note that 
\begin{enumerate}
\item $ \langle \tau-\tau'\rangle \langle \tau'-\tau'' \rangle \leq  \langle \tau-\tau'' \rangle$
\item  The right hand side of \eqref{eq: bound on va norms} contains, through the edge factors $\hat e$, at least one factor $\la^2$ when $0 \notin A$ and at least one factor $\str\la\str$ or $\la^2$ when $0 \in A$.   This is because any contributing fusion has $\caA \neq\emptyset$, see Remark \ref{rem: fusions with zero weight}.   Additional factors  $\e^{C \str A' \str} $ are killed by additional powers of $\la^2$.
\item  The notion of connectedness defined by the relation $\sims$ corresponds to the one on the right hand side of \eqref{eq: bound on va norms} in the following sense:  We start from a fusion $(\caA, \caJ)$ and we choose for any $A' \in \caA$, a spanning tree $\scrT_{A'}$ on $A'$. 
Then, consider the subset of $\scrS$ that consists of $\cup_{A' \in \caA}\caE(\scrT_{A'})$ and of the intervals $J \in \caJ$.  This subset is connected by the adjacency relation $\sims$.
\item  There is at most one edge $E$ containing $0$ so we can absorb an eventual $\breve c(w)$ into the prefactor $\breve C$. 
\end{enumerate}
To finish the proof, we invoke  a combinatorial bound stating that, provided \eqref{eq: kotecky preiss} holds, we have, for any $S_0 \in \scrS$,
\beq  \label{eq: applied cak bound}
\sum_{\substack{\caS \subset \scrS: \caS \sims S_0 \\  \caS \,  \textrm{connected}   }}    \prod_{S \in \caS}   w^{(\be)}_{\mathrm{s}}
(S)   \leq  1,   \qquad     \sum_{\substack{\caS \subset \scrS \\  \caS \cup \{S_0\} \,  \textrm{connected}   }}    \prod_{S \in \caS}   w^{(\be)}_{\mathrm{s}}
(S)   \leq    \e 
\eeq
where $\caS \sims S_0$ means that $S\sims S_0$ for at least one $S \in \caS$. An extended presentation of (a more general version of) this bound is found in Appendix A of \cite{deroeckkupiainenphotonbound}, it is a standard ingredient of cluster expansions.
From  \eqref{eq: applied cak bound}, we  get
\begin{align}
\sum_{\substack{\caS \subset \scrS:\,  \tau \in \supp \caS  \\[1mm]  \caS \,  \textrm{connected}   }}     \prod_{S \in \caS}   w^{(1+\al)}_{\mathrm{s}}
(S)  \leq  C.  \label{eq: bound on v a by graphs}
\end{align}
Indeed, it is straightforward to relate the constraint $\tau \in \supp \caS$ to the adjacency structure defined by $ \sims$. For example: pick an arbitrary $\tau'$ with $\str\tau'-\tau\str \geq 2$ and let $E_{\tau''}$ be the edge $\{\tau',\tau''\}$.  Then,  $\tau \in \supp \caS$ implies that $\caS \sims E_{\tau''}$ for at least one $\tau'' \in \{\tau-1, \tau, \tau+1\}$, and hence \eqref{eq: bound on v a by graphs} follows by the first inequality of \eqref{eq: applied cak bound}.
Combining  \eqref{eq: bound by connected sets} and  \eqref{eq: bound on v a by graphs} yields item 1) and item 2). 
 
Next, we turn to the case where $(n+1)\in A$.  In the simplest case, $A= \{\tau, n+1\}$ for some $\tau$, we have
\beq \label{eq: v bound for singleton}
\str v(\{\tau, n+1\}) \str \leq \breve C \hat e(\tau,\tau\str b)
\eeq 
and all claimed properties, i.e.\ items 3,4,5 follow immediately from properties of $\hat e(\tau,\tau\str b)$.    Let us hence assume that $\str A \setminus \{n+1\}\str >1$ in the remainder of the proof.
Proceeding as in \eqref{eq: bound by connected sets}, we derive
\beq \label{eq: bound by connected sets boundary}
 \e^{c \str A \str}  \str v(A \cup \{n+1\}) \str \leq  \breve C   \sum_{\tau_1,\tau_2\in A, \tau_1\leq\tau_2} \hat e(\tau_1,\tau_2 \str  b)  \sum_{\caS \subset \scrS}   1_{ \supp \caS'=A}  1_{\caS' \,  \textrm{connected} }  \prod_{S \in \caS}   w^{(0)}_{\mathrm{s}}
(S),
\eeq
with $\caS' = \caS \cup \{ \{\tau_1,\tau_2\}  \}$ in case $\tau_1 \neq \tau_2$ and $\caS'=\caS$ if $\tau_1 = \tau_2$. We did not extract $\la^2, \str \la \str$-factors from the right hand side, in contrast to \eqref{eq: bound by connected sets}, because this smallness is anyhow missing in \eqref{eq: v bound for singleton}.   Note furthermore that \eqref{eq: bound by connected sets boundary}  remains valid when we multiply the left hand side by $\dist(A)^{\be}$, and, on the right hand side, we replace $\hat e(\tau_1,\tau_2 \str  b)$ by $  \langle \tau_2-\tau_1\rangle^{\be} \hat e(\tau_1,\tau_2 \str  b)$ and $w^{(0)}_{\mathrm{s}}$ by $w^{(\be)}_{\mathrm{s}}$, for $\be \leq 1+\al$.
To obtain item 4), we use \eqref{eq: bound by connected sets boundary} with these replacements, choosing $\be=1+\al$, and additionally replacing $\str v(\cdot) \str \to \str v(\cdot) \str_{\sup_n}$ and $\hat e(\cdot,\cdot\str b_x)  \to \hat e(\cdot,\cdot\str b_x)_{\sup_n}$. 
  Using the same strategy as in the proof of items 1,2), relying on \eqref{eq: applied cak bound}, we  sum over the collections $\caS$  and over $\tau_1,\tau_2$, using  the bound \eqref{eq: sum edge factors b x} for the edge factor $\hat e(\tau_1,\tau_2 \str  b_x) $. 

To get item 3), we choose $\be=1+\al/2$ and we proceed as previously; the only difference is that we can extract an additional small factor $\delta^{\al/2}$ from the edge factor $\hat e(\tau,\tau' \str b_k)$, i.e.\ we use  \eqref{eq: sum edge factors bmom}.
 
Finally, we deal with item 5).  We abbreviate $ \tilde n:=n-m_\theta n_c$. Note that, if we restrict the sum in \eqref{eqbound scalar polymers mpos special}  to $A$ such that $\max A > \tilde n$, then the desired bound follows from item 4), hence it suffices in the remainder of the proof to restrict the sum to $\max A \leq \tilde n$.  We perform this proof in a more abstract  way than necessary,  because at a later stage we will need an analogous estimate.
We recast the bound \eqref{eq: bound by connected sets boundary} as
\beq \label{eq: rep with x and z}
 \e^{c \str A\str}\str v(A \cup \{n+1\}) \str \leq  \breve C      \sum_{\substack{A_0,A_1 \subset I_{0,n}: \, \str A_0\str =1, 2 \\[1mm] A_0\cup A_1=A  }}   x(A_0)   z_{A_0}(A_1) 
\eeq
where we introduced the weights
\begin{align}
x(A_0) &:=  \hat e(\tau_1,\tau_2 \str b_x), \qquad    A_0=\{\tau_1,\tau_2 \}  \,\,  \text{(possibly $\tau_1=\tau_2 $)},  \\[1mm]
z_{A_0}(A_1) & :=    \sum_{\caS \subset \scrS}   1_{ \supp \caS=A_1}  1_{\caS' \,  \textrm{connected} }  \prod_{S \in \caS}   w^{(0)}_{\mathrm{s}}(S) \label{eq: def z azero}
\end{align}
with  $\caS'$ as in \eqref{eq: bound by connected sets boundary}, and $z_{A_0}(\emptyset):=1$.
The $x,z$-weights  satisfy the properties
\begin{itemize}
\item[$a)$]  
Let $\str a \str_+= \max(a,0)$ for $a \in \bbR$,   
\beq \label{eq: property a}
 \sum_{A_1 \subset I_{0,n} }   \langle \str \min A_0-\min A_1\str_+ \rangle^{1+\al}    z_{A_0}(A_1)  \leq  \breve C \str A_0\str.
 \eeq 
\item[$b)$] 
\beq  \label{eq: property b}
  \sum_{A_0 \subset I_{0,n}: \, \str A_0\str =1, 2, \, \min A_0 \leq \tilde n  } \langle \tilde n -\min A_0 \rangle^{\al}    \e^{ c\str A_0\str}  x(A_0)  \leq  \breve C. \eeq 
\end{itemize}
 Of course, in the case at hand, the right hand side of \eqref{eq: property a}  is simply $\breve C$ by the constraint on $\str A_0\str$. 
Property a) follows by the same reasoning as the proofs of items 1)-4), after writing  the constrained sum over  $A_1$  as $\sum_{\tau \in A_0} \sum_{A_1: A_1 \ni \tau} $, and  property b) is just the bound \eqref{eq: sum edge factors b x special}. The statement of item 5), restricted to $\max A \leq \tilde n$, is now 
\beq \label{eq: sum item five}
 \sum_{\substack{A_0,  A_1 \subset I_{0,n}:\,  \str A_0\str =1, 2,    \\[1mm]      \max (A_0 \cup A_1) \leq \tilde n  } }  \langle \tilde n -\min (A_0 \cup A_1)\rangle^{\al}  x(A_0)   z_{A_0}(A_1) \leq \tilde C.
\eeq
Note that
\begin{align}
\langle \tilde n -\min (A_0 \cup A_1)\rangle  
&\leq    C \langle  \tilde n -\min A_0\rangle  \langle \str\min A_0 -\min A_1 \str_+\rangle. 
\end{align}
We substitute this in the left hand side of \eqref{eq: sum item five} and use property a) to bound the sum over $A_1$ by $C\str A_0\str \leq \e^{c \str A_0\str } $. Then, we  perform the sum over $A_0$ by property b). This proves the inequality \eqref{eq: sum item five}.  \qed

\section{Proofs of the main theorems}
In this Section, we give the final proof of our main results, Theorems \ref{thm: propagation estimate} and \ref{thm: soft photon bound}. First, in Section \ref{sec: general considerations}, we introduce some general tools, applying to both choices of the operator $b$. Most importantly, we develop a refinement of the representation \eqref{eq: basic rep zn}.   In Section \ref{sec: propagation bound}, we specialise to the case $b=b_x$ and we prove the minimal velocity estimate, i.e.\ Theorem \ref{thm: propagation estimate}. In Section \ref{sec: soft photon bound}, we take $b=b_k$ and we obtain the soft boson bound, i.e.\ Theorem \ref{thm: soft photon bound}. 

\subsection{General Tools}\label{sec: general considerations}

As announced, we do not distinguish for now between the two different choices  for $b$, except in Lemma \ref{lem: bounds barred weights}.
We start from the representation \eqref{eq: basic rep zn} 
and we introduce some notation to simplify it.  We will use the adjacency relation $A \sim A'   \Leftrightarrow \distance(A,A') \leq 1$  for subsets of $I_{0,n}$, and extended to subsets of $I_{0,n+1}$ by simply ignoring the element $n+1$, i.e.:
  \beq
  A \sim A'   \Leftrightarrow \distance(A \setminus \{n+1\},A'  \setminus \{n+1\}) \leq 1,\qquad   A,A'  \neq \{n+1\},
  \eeq
 (we never need the case where $A$ or $A'$ is the singleton $\{n+1\}$). As previously, we write $\caA \sim A'$ if there is  at least one $A  \in \caA$ such that $A\sim A'  $, and $\caA \nsim A'$ if there is no $A \in \caA$ such that $A \sim A'$.
 
We recast \eqref{eq: basic rep zn} 
by separating each collection $\caA$ into its boundary and bulk polymers; 
 \beq \label{eq: splititng of zn into boundary parts}
Z_{n} =   \sum_{A_\realfinal} v(A_\realfinal) Z_{n, A_\realfinal} +    \sum_{A_\realfinal, A_\realinitial:  A_\realfinal \not\sim A_\realinitial} v(A_\realinitial) v(A_\realfinal)   Z_{n, A_\realfinal \cup A_\realinitial}   +  \sum_{A_{\realinitial,\realfinal}}   v(A_{\realinitial,\realfinal})   Z_{n,A_{\realinitial,\realfinal}}  \eeq
where we abbreviated $Z_n=Z_n(\d\Ga(b), \rho_0)$ and where  $A_\realinitial, A_\realfinal, A_{\realinitial,\realfinal}$ run over nonempty subsets of $I_{0,n+1}$ that, respectively, 
\begin{itemize}
\item  contain $0$ but not $n+1$,
\item contain $n+1$ and at least one other element, but not $0$.
\item contain both $0$ and $n+1$.
\end{itemize}
and the factors $Z_{n,A}$ in \eqref{eq: splititng of zn into boundary parts} are defined as
\beq   Z_{n, A'}  :=
  \sum_{\scriptsize{\left.\begin{array}{c}   \caA \in {\frB}^1_{1,n}     \\   \caA \nsim  A' 
  \end{array} \right. }}    \prod_{A \in \caA}  v(A)      \label{eq: def z restricted}
\eeq
where it is understood that $\caA=\emptyset$ contributes  $1$ to the right hand side. 
Note that $Z_{n, A}$ depends only on bulk polymer weights.  Moreover, by \eqref{eq: z only bulk},
\beq  \label{eq: definition empty zet}
1= Z_n(\lone, \rho_0) =   Z_{n , \emptyset}, \qquad \textrm{for}\,\,  \rho_0 = \eta \otimes P_\Om.
\eeq  
As explained in \cite{deroeckkupiainenphotonbound}, 
the quantity $Z_{n, A'}$ can be viewed as the partition function of a polymer gas with  polymer weights  $
w (A) \equiv  v(A)   1_{[A \nsim  A']} $. 
For $\la$ small enough, the bound \eqref{eqbound scalar polymers} (a 'Kotecky-Preiss' criterion, in the terminology of \cite{deroeckkupiainenphotonbound}) allows us to apply the cluster expansion and obtain
\beq \label{eq: rep z n a}
\log Z_{n, A'}    =  \sum_{\caA \in \frB_{1,n}}   v^T(\caA)      \indicator_{[ \caA \nsim A' ]} 
\eeq
where the \emph{truncated} weights $v^T(\cdot)$ are defined as
  \beq  \label{def: truncated cluster weights}
v^T(\caA) :=  \sum_{\scrG \in \frG^c(\caA) }  (-1)^{\str  \scrE(\scrG) \str}    \prod_{\{A_{i}, A_j\} \in  \scrE(\scrG)} 1_{[A_{i} \sim A_j]}  \prod_{A_{i} \in \caA} v(A_i)
\eeq
where $\frG^c(\caA)$ is the set of connected graphs with vertex set $\caA$, and 
 $\scrE(\scrG) $ is the edge set of the graph $\scrG$, see Appendix A of \cite{deroeckkupiainenphotonbound} for more details.  The only property of the weights $v^T(\cdot)$ that we need here is\footnote{The property stated in  Appendix A of \cite{deroeckkupiainenphotonbound} misses the factor $\e^{c \str \supp\caA\str}$ but this can be easily obtained by redefining $v(A) \to \e^{c \str A \str} v(A)$ and taking $\str \la \str$ smaller.}
\beq  \label{eq: bound z decay}
\sum_{\caA \in \frB_{1,n}:\, \caA \sim A}      \dist(\caA)^{1+\al}  \e^{c \str \supp \caA \str}  \str v^T(\caA) \str        \leq  C\str \la^2 \str \str A \str.
\eeq
Comparing to the expansion of $Z_{n, \emptyset}$ and using $\log Z_{n, \emptyset}=0$ (see \eqref{eq: definition empty zet}), we get
\beq \label{eqgeneral expression for excluded z}
\log Z_{n,A'}  =   \log \frac{Z_{n,A'} }{Z_{n, \emptyset} }   =   -\sum_{\caA \in \frB_{1,n}}  v^T(\caA)       \indicator_{[ \caA \sim A' ]}. 
\eeq
We now decompose 
\beq \label{eq: decomposition of zet}
Z_{n,A'}=   \sum_{A \subset I_{1,n}} p_{A'}(A)
\eeq
with $p_{A'}(\emptyset)=1$ and, for $A \neq \emptyset$, 
\beq  \label{eq: weights p}
 p_{A'}(A)=  \sum_{  \substack{ \frA \subset  \frB_{1,n}   \\[0.5mm]   \supp \frA=A }} \mathfrak{p}_{A'}(\frA), \qquad  
   \mathfrak{p}_{A'}(\frA)   =    \prod_{ \caA \in \frA } (\e^{-v^T(\caA)}-1)  1_{  \caA \sim A'}
\eeq
with  $\supp \frA= \cup_{\caA \in \frA} \supp \caA$.  The decomposition \eqref{eq: decomposition of zet} follows from the identity
$$\prod_{x \in X} \e^{f(x)} = \sum_{Y \subset X}\prod_{x \in Y} (\e^{f(x)}-1), $$
 for a finite set $X$ and $f: X \to \bbC$ and with $\prod_{x \in \emptyset} := 1$. 
 
Next, we simplify \eqref{eq: splititng of zn into boundary parts} by introducing new weights $\bar v(\cdot)$.  In what follows, $A_1 $ ranges over subsets of $I_{1,n}$, $A_\realinitial, A_\realfinal, A_{\realinitial, \realfinal}$ have the same meaning as before in \eqref{eq: splititng of zn into boundary parts}. 
First we define 
\beq
\bar v^{(1)}(A_{\sharp}) :=
 \sum_{A'_{\sharp},A_1:  A'_{\sharp} \cup A_1=A_\sharp
 }     v(A'_{\sharp})  p_{A'_{\sharp}}(A_1)     \label{def: bar v 1}
\eeq
where $({\sharp})$ stands for either one of the three subscripts $({\realinitial}),({ \realfinal}), ({ \realinitial, \realfinal})$, the same subscript on the left and right hand side of the equation.    Then, we also need 
\begin{align}
& \bar v^{(2)}(A_{\realinitial})  :=0,  \qquad \qquad  \bar v^{(2)}(A_{\realfinal}) :=0, &   \\[1mm] 
& \bar v^{(2)}(A_{\realinitial,\realfinal})  :=
 \sum_{ \substack{A_\realinitial, A_\realfinal,   A_1   \\   A_\realinitial \nsim A_\realfinal, A_\realinitial \cup A_\realfinal \cup A_1=A_{\realinitial,\realfinal}
 } }   v(A_\realinitial) v(A_\realfinal) p_{ A_\realinitial \cup A_\realfinal}(A_1)  &    \label{def: bar v 2}
\end{align}
and  finally
\beq
\bar v^{}(A_{\sharp}):= \bar v^{(1)}(A_{\sharp}) +\bar v^{(2)}(A_{\sharp}).
\eeq
Moreover, we define again $\bar v(\{n+1\}):=v(\{n+1\})= 0$. 
Relying on \eqref{eq: decomposition of zet}, we recast \eqref{eq: splititng of zn into boundary parts}  as  
\beq  \label{eq: basic rep}
Z_n = \sum_{A_\realfinal }   \bar v(A_{\realfinal}) + \sum_{A_\realinitial, A_\realfinal: A_\realinitial \nsim A_\realfinal }   \bar v(A_{\realinitial})  \bar v(A_{\realfinal}) +    \sum_{A_{\realinitial,\realfinal}}   \bar v(A_{\realinitial,\realfinal}). 
\eeq
For example, note that the $\bar v^{(2)}(\cdot)$ weights account for contributions to the second term of \eqref{eq: splititng of zn into boundary parts} that contribute to the third term in \eqref{eq: basic rep}.  In other words, if we expand  $Z_{n,A_\realinitial \cup A_\realfinal}$ in the second term of \eqref{eq: splititng of zn into boundary parts} according to \eqref{eq: decomposition of zet}, then the terms with $A \sim A_\realinitial,  A \sim A_\realfinal $ contribute to the $\bar v^{(2)}(\cdot)$ weights. 

Furthermore, we rewrite  \eqref{eq: basic rep} by first remarking that (for any $n$)
\beq \label{eq: sum initial zero}
\sum_{ A_{\realinitial} } \bar v(A_{\realinitial}) =0. 
\eeq
Indeed, consider  \eqref{eq: basic rep}  for $Z_n(\lone, \rho_0)=1$, then  polymers $A$ with $n+1 \in A$ never appear and in that case  \eqref{eq: basic rep} simply reads  $1= 1+\sum_{ A_{\realinitial} } \bar v(A_{\realinitial})   $. 
Then, we  decompose $\sum_{A_\realinitial \nsim A_\realfinal} = (\sum_{A_\realinitial } )(\sum_{A_\realfinal} ) - \sum_{A_\realinitial \sim A_\realfinal} $  so that we get our final expression
\beq  \label{eq: zn three terms general}
 Z_n =  \sum_{A_\realfinal } \bar v(A_\realfinal)  -  \sum_{ A_\realinitial \sim A_\realfinal } \bar v(A_{\realinitial})  \bar v(A_{\realfinal}) +    \sum_{A_{\realinitial,\realfinal} } \bar v(A_{\realinitial,\realfinal}).  
\eeq 

The weights $\bar v(\cdot)$ have analogous properties to the   $v(\cdot)$-weights. The time-translation invariance properties will be stated later, here we deal with the bounds:
\begin{lemma} \label{lem: bounds barred weights}
Items $2,3,4,5$ of Lemma \ref{lem: bound on scalar polymers} hold with $v(\cdot)$ replaced by $\bar v(\cdot)$,  possibly with different constants $c, C,\breve C$. 
\end{lemma}
We will henceforth refer to items $2,3,4,5$ of Lemma \ref{lem: bounds barred weights} (There is no item 1) since we did not define $\bar v(A)$  for $A \subset I_{1,n}$).
\begin{proof}
The $\bar v$-weights are built from the $v$-weights by `dressing' - in the sense of (\ref{def: bar v 1}, \ref{def: bar v 2}) - the $v$-weights with bulk polymers whose $p$-weights are small and have strong summability properties,  as Lemma \ref{lem: bound on fras} below shows. This should be compared to the proof of Lemma \ref{lem: bound on scalar polymers}  where the edge factors $\hat e(\tau_1, \tau_2\str b)$ were dressed with collections $\caS$, cfr.\ \eqref{eq: bound by connected sets boundary}.  
We first abbreviate
\beq
  \frp_{A'}(\caA)=  (\e^{-v^T(\caA)}-1)  1_{  \caA \sim A'}, \qquad \text{such that} \,    \mathfrak{p}_{A'}(\frA)   =    \prod_{ \caA \in \frA }   \frp_{A'}(\caA)
  \eeq
and
\beq
r(A) = d(A)^{1+\al}  \e^{c \str A \str}, \qquad r( \caA) =  r(\supp \caA). 
\eeq
Then
\begin{lemma}\label{lem: bound on fras}
For any $A' \neq \emptyset$, 
\beq \label{eq: prop of p weights}
  \sum_{\frA \subset \frB_{1,n} }      \prod_{\caA \in \frA} r(\caA) \str  \frp_{A'}(\caA) \str    \leq   \e^{C \lambda^2   \str A' \str}.  
\eeq
\end{lemma}
\begin{proof}
We bound the left hand side of \eqref{eq: prop of p weights} by
\beq
1+   \sum_{k=1}^{\infty} \frac{1}{k!}  \sum_{\caA_1, \ldots, \caA_k \in \frB_{1,n}  }       \prod_{j=1}^k  r(\caA_j)  \str   \frp_{A'}(\caA_j) \str  \leq  \e^{\sum_{\caA \in \frB_{1,n} }    r(\caA)  \str   \frp_{A'}(\caA) \str }.  \nonumber
\eeq
The exponent is bounded by $C\la^2 \str A' \str$, by \eqref{eq: bound z decay}. 
\end{proof}
We give sketches of the proofs of items 3) and 5). The remaining items 2) and 4) are treated similarly to 3) and we omit their proofs. 
We start with item 3).   For the sake of simplicity we drop the  $\bar v^{(2)}(\cdot)$ contribution as its treatment is similar to the $\bar v^{(1)}(\cdot)$ contribution. 
To get item 3) with $\bar v(\cdot)$ replaced by $\bar v^{(1)}(\cdot)$, it suffices to show
\beq
\sum_{\substack{A_0 \subset I_{0,n}, \frA \in \frB_{1,n} \\[0.5mm]   \tau \in A_0\cup \supp \frA  }}  \big(\prod_{\caA\in \frA} r( \caA)  \str p_{A_0}(\caA )\str\big) \,   r(A_0) \str v(A_0 \cup \{n+1\} )\str   \leq \breve C.
\eeq
We dominate this as $\sum_{A_0, \frA:  \tau \in A_0\cup \supp \frA} \leq \sum_{ A_0:  \tau \in A_0 }+ \sum_{\frA: \tau \in  \supp \frA } $.   In the first term, we first estimate the sum over $\frA$ by $\e^{C \la^2 \str A_0\str}$ using Lemma \ref{lem: bound on fras}, and then we use item 3) of Lemma \ref{lem: bound on scalar polymers}  to perform the sum over $A_0 $ with $ \tau \in A_0$.  In the second term, we pick arbitrarily a $\caA \in \frA $ such that $\tau \in \supp \caA$ (hence in particular $\caA \sim A_0$ and $\caA \sim \{\tau\}$)   and we dominate this term by
\beq
\sum_{\caA \in \frB_{1,n}}  \str   \frp_{\{\tau\}}(\caA)  \str    r( \caA)  \sum_{A_0 \subset I_{0,n}: A_0 \sim \caA}  r(A_0)  \str   v(A_0 \cup \{n+1\} ) \str     \sum_{\frA' } \prod_{\caA' \in \frA'}  r(\caA')  \str  \frp_{A_0}(\caA')  \str.  
\eeq
The sum over $\frA'$ is dominated by $\e^{C \la^2 \str A_0 \str}$ by Lemma \ref{lem: bound on fras}, the sum over $A_0$ is dominated by $\breve C \str \supp \caA\str \leq \breve C \e^{c \str \supp \caA\str}$ by item 3) of Lemma \ref{lem: bound on scalar polymers} upon adjusting $c$, and the final sum over $\caA$ is dominated by $\breve C$  by \eqref{eq: bound z decay}, again adjusting $c$. 
Finally, we treat item 5). As argued in the proof of item 5) in Lemma \ref{lem: bound on scalar polymers}, we can assume $\max  A \leq \tilde n$ (the case $\max  A > \tilde n$ being handled by item 4)), hence it suffices to show
\beq \label{eq: sum item five again}
 \sum_{\substack{A_0 \in I_{0,n},  A_1 \subset I_{1,n} \\[0.5mm]   \max (A_0 \cup A_1) \leq \tilde n} }  \langle \tilde n -\min (A_0 \cup A_1)\rangle^{\al}  x(A_0)   z_{A_0}(A_1) \leq \tilde C
\eeq
with $\tilde n= n-n_c m_\theta$ and 
$$x(A_0):= \e^{c \str A_0\str} \str v(A_0 \cup \{n+1\})\str, \qquad    z_{A_0}(A_1):= \e^{c \str A_1\str} \str p_{A_0}(A_1)\str.$$ 
Note the similarity of \eqref{eq: sum item five again} with \eqref{eq: sum item five}, the only difference being that here we do not restrict $\str A_0\str$ and that we have $A_1 \subset I_{1,n}$, i.e.\ $ 0 \notin A_1$.  With these small changes, the properties $a),b)$ (\ref{eq: property a}, \ref{eq: property b}) hold with the $x,z$ weights as defined here: Property  a) by \eqref{eq: weights p} and Lemma \ref{lem: bound on fras}, and property b) by  item 5) of Lemma \ref{lem: bound on scalar polymers}.   Therefore, we can repeat the short proof given in the proof of item 5) of Lemma \ref{lem: bound on scalar polymers} to get the desired claim. 
\end{proof}

\subsection{Symmetry properties of the $\bar v(\cdot)$ weights} \label{sec: symmetry of weights}

We list some symmetry properties of the $\bar v(\cdot)$ weights. We indicate the dependence on the final time explicitly by writing $
\bar v_n(A) $  instead of  $\bar v(A)$.

Let first  $b=b_k$. For $\tau \in \bbN$, $n'>n$ and $A$ such that  both $A,A+\tau$  are subsets of $ I_{1,n}$
\beq \label{eq: invariance of barred v mom} 
\bar v_{n}((A+\tau)\cup \{n+1\})    = \bar v_{n}(A\cup \{n+1\})  =  \bar v_{n'}(A\cup \{n'+1\}). 
\eeq
Let now $b=b_x$, then these equalities do not hold in general, but we still have 
\beq \label{eq: invariance of barred v pos} 
\bar v_{n}(A\cup \{n+1\})   =    \bar v_{n+\tau}((A+\tau)\cup \{n+\tau+1\})  
\eeq
where it is understood that $n_c$ is kept fixed. 
To establish these properties, one first checks that the same properties hold for the $v(\cdot)$ weights.  This follows easily from the symmetry properties in Section \ref{sec: symmetry prop}.  Since (unlike the $\bar v(\cdot)$ weights) the $v(\cdot)$ weights are also defined for $A \subset I_{1,n}$, we can also state an additional symmetry property, namely, for $\tau \in \bbN$, $n'>n$ and $A$ such that  both $A,A+\tau$  are subsets of $ I_{1,n}$,
\beq \label{eq: invariance of v} 
 v_n({A+\tau}) =  v_n({A})  =   v_{n'}({A}). 
\eeq
Finally the symmetry properties of the $\bar v(\cdot)$ weights follow from the corresponding properties for the $v(\cdot)$ weights and from \eqref{eq: invariance of v}.

\subsection{Minimal velocity estimate} \label{sec: propagation bound}
In this section, we take throughout $b=b_x= \theta(x/t_c)$ and we again abbreviate $Z_n=Z_n(\d \Gamma(b_x),\rho_0)$.  Assume the same conventions for the sets $A_{\realinitial}, A_{\realfinal}$ as above, then we have
\begin{lemma}\label{lem: finite size corrections zn} 
\beq
\left \str Z_n -   \sum_{A_\realfinal } \bar v(A_\realfinal)  \right\str \leq  \breve C \langle n \rangle^{-\al}.    \label{eq: finite size zn} 
\eeq 
\end{lemma}
\begin{proof}
 From \eqref{eq: zn three terms general},  the expression between $\str \cdot \str$ equals
\beq \label{eq: error terms}
-\sum_{ A_\realinitial \sim A_\realfinal }  \bar v(A_{\realinitial})  \bar v(A_{\realfinal}) +    \sum_{A_{\realinitial,\realfinal}}     \bar v(A_{\realinitial,\realfinal}). 
\eeq
Let us treat the first term.   We distinguish  the cases $\min A_\realfinal \leq (n-m_\theta n_c)/2$ and  $\min A_\realfinal > (n-m_\theta n_c)/2$. 
In the first case, we first sum over $A_\realinitial$ using item
 2) of Lemma \ref{lem: bounds barred weights}, yielding a factor $\breve C$ (in fact $\breve C \str \la \str$) and then over $A_\realfinal$, using item
 5) of Lemma \ref{lem: bounds barred weights}  and the fact that 
$$ n-n_cm_\theta-(n-m_\theta n_c)/2 \geq (1-m_\theta)n/2 = \breve c n, $$ 
and obtaining 
$\breve C \langle n \rangle^{-\al} $ (recall that $n_c \leq n$).  In the second case, we first sum over $A_\realinitial$ (item
 2) of Lemma \ref{lem: bounds barred weights})  obtaining a  factor $\breve C \langle \min A_{\realfinal} \rangle^{-(1+\al)}$ because $\max A_\realinitial \geq \min A_\realfinal-1$ (since $A_\realinitial \sim A_\realfinal$),  then we sum over $A_\realfinal$ keeping $\min A_\realfinal$ fixed, yielding $\breve C$ by item
 4) of  Lemma \ref{lem: bounds barred weights}, and  finally over $\min A_\realfinal$,  yielding $ \breve C \langle  (n-m_\theta n_c)/2 \rangle^{-\al}=\breve C \langle  n \rangle^{-\al}$. The second term in \eqref{eq: error terms} is estimated in an analogous way.
\end{proof}

Now we are ready to consider the limit $n \to \infty$ in the expression for $Z_n(\d\Gamma(b_x),\rho_0)$.  Therefore, we should render the $n$-dependence in the weights explicit, as we did in Section \ref{sec: symmetry of weights}. Therefore,  we introduce new notation, namely, for a finite $A \subset \bbN_0$,
\baq  \label{eq: def nu weights}
\nu(A):= \bar v_n((n+1-A) \cup \{n+1\} ),  \qquad \text{with $n$ such that $A \subset I_{1,n}$}.
\eaq
With this notation,  we have (recall that the sum over $A_{\realfinal}$ runs over sets not including $0$, and hence $\bar v(A_\realfinal)$ does not depend on $\rho_0$)
\beq \label{eq: new notation}
 \sum_{A_\realfinal } \bar v(A_\realfinal) =   \sum_{A \subset I_{1,n} } \nu(A). 
\eeq
Let us define
\beq \label{def: z infty}
Z_{\infty}: =  \sum_{A \subset \bbN :\, \str A \str < \infty}       \nu(A).   
\eeq
Note that $Z_{\infty}$ still depends on $t_c$ (or $n_c$), but not on $\rho_0$. 
\begin{lemma}   \label{lem: z infity is limit}  
This sum on the right hand side of \eqref{def: z infty} is absolutely convergent and 
\begin{eqnarray} 
 \left \str Z_n -Z_{\infty} \right \str &\leq &  \breve C \langle n \rangle^{-\al}.    \end{eqnarray}
\end{lemma}
\begin{proof} 
We have
\begin{align}
 \big\str Z_{\infty}-    \sum_{A \subset I_{1,n} }   \nu(A)    \big\str  &\leq  \limsup_{n' \to \infty}  \sum_{A \subset I_{1,n'}:\, \min A < (n'-n)+1} \str \bar v_{n'}(A \cup  \{n'+1\})\str \\
&\leq   \breve C \langle n \rangle^{-\al}.& \label{eq: conv to zrealfinal}
 \end{align}
The first inequality is by \eqref{eq: def nu weights}, the second is by   item 5) of Lemma \ref{lem: bounds barred weights} since $Cn \leq (n'-m_\theta n_c-\min A) $ for any $A$ contributing to  the sum. 
The Lemma then follows by the triangle inequality from (\ref{eq: conv to zrealfinal},\ref{eq: finite size zn})  and \eqref{eq: new notation}. 
\end{proof}

\subsubsection{Proof of  Theorem \ref{thm: propagation estimate}}  \label{sec: proof of thm propagation}

First, we slightly  generalise the setup;  we consider $Z_n(O, \rho_0)$ with $O=\d\Ga(b_x), \lone $ and  $\rho_0$ now not longer a density matrix but the rank-1 operator
\beq
\rho_0= \big\str  \psi_{\sys}\otimes \caW(\psi_{\realinitial})\Om \big\rangle \big\langle   \psi'_{\sys}\otimes \caW(\psi_{\realinitial}')\Om \big\str
\eeq
with $\psi_\sys,\psi_\sys' \in \scrH_\sys$ and $\psi_\realinitial, \psi_\realinitial' \in \frh_\al$.  
 We can easily go through all arguments, with obvious changes, and get the following analogue of Lemma \ref{lem: z infity is limit} 
\beq \label{eq: generalized z conv}
\left \str Z_n(\d\Ga(b_x), \rho_0)  -  Z_{\infty}  \Tr \rho_0 \right \str \leq  \breve C \langle n \rangle^{-\al} ,   \qquad Z_n(\lone, \rho_0) =  \Tr \rho_0 
\eeq
with $Z_\infty$ as above in \eqref{def: z infty}.
Now, take  $\Psi \in \caD_\al$, i.e.\ a finite linear combination $\Psi=\sum_{i}\Psi_{i}$ with $\Psi_i= \psi_{\sys,i}\otimes \caW(\psi_{\realinitial,i})\Om$, then 
\beq \label{eq: conv cadalpha}
\lim_{n \to \infty} \langle    \Psi
(n/\la^2), \d\Ga(b_x)   \Psi
(n/\la^2) \rangle   = Z_{\infty} \sum_{i,j} \langle \Psi_i , \Psi_j \rangle =  Z_{\infty} \norm \Psi \norm^2.
\eeq
Hence,  we  get the statement of Theorem   \ref{thm: propagation estimate}  for times $t$ taken along a subsequence $n/\la^2$ and $t_c$ of the form $n_c/\la^2$. To get the full statement, we should again generalise the reasoning in a straightforward way. 

Assume that the time-discretisation of the model was chosen based on 'mesoscopic time-blocks' of length $\ell \str\la\str^{-2}, \ell \in [1,2]$, instead of $\ell=1$ as we did previously: this means that we change the definition of $Q_n,{Q}_{n \str b}$ and $U_{\tau}, \tau=1,\ldots,n$ by replacing $\str\la\str^{-2}$ by $\ell \str\la\str^{-2}$, for example, instead of \eqref{def: u tau}, we have
\beq \label{def: u tau new}
U_{\tau} :=      \e^{\i \tau (\ell /\la^2)  L_\res} \e^{-\i (\ell/\la^2)  L} \e^{-\i (\tau-1)(\ell/\la^2)  L_\res}   ,\qquad \tau \in I_{1,n}.
\eeq
Then, Lemma \ref{lem: spectral gap} holds as well with a constant $C^{(\ell)}$ and gap $g^{(\ell)}$ that can be chosen uniform in $\ell \in [1,2]$, as we easily get from the results in \cite{deroeckkupiainenphotonbound}, in particular from the proof of Lemma 2.3 1) therein.  The rest of the reasoning goes through without any change except for the readjusting of constants. Hence we have now proven  Theorem   \ref{thm: propagation estimate}  restricted to times $t$ taken along a subsequence $n \ell/\la^2$ and $t_c$ of the form $n_c\ell/\la^2$, and with constants $\breve C$ on the right hand side that can be chosen uniform in $\ell \in [1,2]$. 
Finally, $t_c$ can be tuned independently of $t$ by changing the function $\theta(\cdot)$ to $\theta(\ell\cdot)$ for  $\ell \in [1,2]$ and again the constants $\breve C$ can be chosen uniform. This allows to choose any $t_c \geq \la^{-2} $  (smaller $t_c$ would require to take $\ell$ dependent on $\la$ which we prefer to avoid) and to  establish the full Theorem   \ref{thm: propagation estimate}.

\subsection{Soft boson bound}\label{sec: soft photon bound}
In this section, we take to $b=b_k= \theta(k/\delta)$.  Recall the conventions for $A_\realinitial, A_\realfinal, A_{\realinitial,\realfinal}$ and the expression \eqref{eq: zn three terms general}:
\beq  \label{eq: zn three terms mom}
 Z_n =  \sum_{A_\realfinal } \bar v(A_\realfinal)  - \sum_{ A_\realinitial \sim A_\realfinal } \bar v(A_{\realinitial})  \bar v(A_{\realfinal}) +    \sum_{A_{\realinitial,\realfinal} } \bar v(A_{\realinitial,\realfinal}).  
\eeq 
The second and third term on the right hand side are bounded by $ \breve C \delta^{\al/2}$, using items 2) and 3)  of Lemma \ref{lem: bounds barred weights}. 
For the first term on the right hand side, we argue 
\begin{lemma} \label{lem: extensive soft photon}
There is an $n$-independent number $a$ such that 
\beq \label{eq: finite size k}
\big\str\sum_{A_\realfinal } \bar v(A_\realfinal) - n a  \big\str \leq \breve C \delta^{\al/2}.
\eeq
\end{lemma} 
\begin{proof}
To deal with the $n$ -(in)dependence of the weights, we again introduce new notation:
\beq
\bar v(A\str b_k) :=  \bar v_{n}(A \cup \{n+1\}), \qquad   \text{where $n\geq \max A$},
\eeq
and, by  \eqref{eq: invariance of barred v mom} we  have $\bar v(A\str b_k)=\bar v(A+\tau \str b_k) $ provided both $A,A+\tau$ are finite subsets of $\bbN_0$. 
Then, let us define 
\beq
a := \sum_{A \subset \bbN_0:\,  \str A\str <\infty, \min A=1}   \bar v(A\str b_k)
\eeq
where the sum  is absolutely convergent by item
 3) of Lemma \ref{lem: bounds barred weights}, and we have
\beq
na- \sum_{A \subset I_{1,n}}   \bar v(A\str b_k) =  \sum_{A \subset \bbN, \min A=1}    \min (\max A-1,n)\bar v(A\str b_k).
\eeq
The left hand side is the expression between $\str \cdot \str$ in \eqref{eq: finite size k}, and the right hand side  can be bounded by  $\breve C \delta^{\al/2}$ by using again item
 3) of Lemma \ref{lem: bounds barred weights}. 
\end{proof}

By the boson number bound in \cite{deroeckkupiainenphotonbound}, we know that $\sup_t \langle \Psi_t, N \Psi_t \rangle  \leq \breve C$, and therefore $ \sup_n Z_n(\d \Ga(b_k),\rho_0) \leq \breve C$, see the remark following Theorem \ref{thm: soft photon bound}. However,  Lemma \ref{lem: extensive soft photon} and the bounds on the other terms (second and third) of \eqref{eq: zn three terms mom} imply that  $Z_n(\d \Ga(b_k),\rho_0)-an $ is uniformly bounded in $n$. Combining  these two statements, we conclude $a=0$, and therefore, we have shown 
\beq
\str Z_n(\d \Ga(b_k),\rho_0) \str \leq \breve C \delta^{\al/2}.
\eeq
We have hence obtained Theorem \ref{thm: soft photon bound} for $t$ restricted to particular vectors $\Psi_0= \psi_{\sys}\otimes \caW(\psi_\realinitial)\Om$ and times of the form $t= n/\la^2$.  By the same trick as applied at the end of the proof of Theorem \ref{thm: propagation estimate} in Section
\ref{sec: proof of thm propagation} (involving the change of mesoscopic scale $\str\la\str^{-2} \to \ell \str\la\str^{-2}$), we get the statement for all times $t$. By the Cauchy-Schwarz inequality, we get  the statement for any $\Psi \in \caD_\al$.   This proves the full  Theorem \ref{thm: soft photon bound}.

\appendix
\renewcommand{\theequation}{\Alph{section}\arabic{equation}}
\setcounter{equation}{0}

 \section{One-particle estimates}  
 \label{app: propagation estimates}

\renewcommand{\theequation}{A-\arabic{equation}}
  \setcounter{equation}{0}  

In this appendix we prove some estimates concerning the dynamics of a single free boson.  Very similar estimates were also established by different methods in \cite{gerardscatteringmasslessnelson} (the approach here is less elegant, but more self-contained)

 Since this section stands apart from the rest of the paper, we we do not rely on previous definitions and conventions, unless explicitly mentioned.  In particular, we do not adhere to our earlier convention to distinguish constants $C$ and $\breve C$. 
 First, we state
 \begin{lemma} \label{lem: decay}
 For $f \in L^2(\bbR^d)$, assume that  $\hat f  \in C^1(\bbR^d \setminus \{0\})$  such  that  for some  $0<\gamma<1$,
\beq  \label{eq: ga decay}
 |k|^{1-\gamma} \partial \hat f  \in L^1(\bbR^{d}; \bbC^d), \qquad   |k|^{-\gamma}\hat f \in L^1(\bbR^{d}).
  \eeq 
Then,  
$$
\left | f(x)\right |\leq   C(\gamma)   \str x\str^{-\gamma}
  \left(  \norm   \str k\str^{-\gamma}\hat f \norm_1 + \norm    \str k\str^{1-\gamma} \partial \hat f \norm_1 \right).  $$
\end{lemma}
\begin{proof}
We write
$$
f(x)=\frac{1}{\e^{-\i}-1} \int \d k\, ({\hat f(k +\hat x/|x|)-\hat f(k )}) 
\e^{\i kx }.
$$
Divide the integral to $|k|\leq 2 |x|^{-1}$ and  $|k|> 2 |x|^{-1}$. For the first one 
insert $1\leq 2^\gamma |x|^{-\gamma}|k|^{-\gamma}$
and the integral  is bounded by $2^{1+\gamma} |x|^{-\gamma} \norm   \str k\str^{-\gamma}\hat f \norm_1$. For the second 
integral, we can assume that $\hat f $ is $C^1$, hence we
insert 
$$\hat f(k +\hat x/|x|)-\hat f(k )=|x|^{-1}\int_0^1\d s\, \hat x\cdot \partial\hat f(k+s\hat x/|x|).
$$
to bound it by
\baq
\int_0^1 \d s\int_{|k|>\frac{2}{ |x|}} \d k\,  |x|^{-1}
|\partial \hat f(k +s\hat x/|x|)|&\leq&
 |x|^{-\gamma}\int_{|k|> \frac{1}{ |x|}} \d k\,  |x|^{-1+\gamma}
 |\partial \hat f(k )|\nonumber\\
&\leq&C |x|^{-\gamma} \norm    \str k\str^{1-\gamma} \partial \hat f \norm_1.
\eaq
since   $|x|^{-1+\gamma}\leq |k|^{1-\gamma}$  in the last integral.
\end{proof}

\subsection{Minimal velocity estimates}
Recall the  function $\theta$ introduced above Theorem \ref{thm: soft photon bound}. It is a spherically symmetric $C^{\infty}$ function $\bbR^d \to [0,1]$ with support contained in a ball with radius $r_\theta <1$ and we write $\theta_s(x):= \theta(x/s)$.  Recall also the dense subspace $\frh_\al \subset L^2(\bbR^d)$ and write $\psi_s=\e^{-\i \om s} \psi$. 
 We prove 
 \begin{lemma}\label{lem: bounds g}  Let $d\geq 3$ and  $\psi, \psi' \in \frh_\al$. There is a $\ga>\al$ such that, for any $s>0$ and 
 $1> m_\theta > r_\theta$,
     \beq 
  \str (\psi'_{s_2},\theta_s ({x})\psi_{s_1}) \str   \leq C \langle s_2-s_1 \rangle^{-2-\ga},     \label{eq: standard bound h app} \eeq 
     \beq   \int_{ s m_{\theta}}^{\infty}  \d s_2 \int_{0}^{s_2}   \d s_1   \,    \langle s_2-m_{\theta} s \rangle^{\al}  \str (\psi'_{s_2},\theta_s ({x})\psi_{s_1}) \str    \leq C
   \label{eq: nonstandard bound h app} \eeq 
   where $C$  depends on $\ga$,  $\theta$ and $m_\theta$ (in particular it diverges when $m_\theta \to r_\theta $), but not on $s$. 
   \end{lemma}
Upon renaming the time variables, 
Lemma \ref{lem: bounds g} yields the claims of Proposition \ref{prop: bounds on correlation functions} with $b=b_x$: 
Item
 4) follows from the bound \eqref{eq: standard bound h app}  by choosing $\psi,\psi'$ either $\phi$ or $\psi_\realinitial$, $s=t_c$ and noting that, for example,
 $\langle \e^{\i v \om}\phi, b_x(t) \e^{\i u \om}\phi\rangle$ is the complex conjugate of  $\langle \e^{-\i (t-u) \om}\phi, \theta(x/t_c)\, \e^{-\i (t-v) \om}\phi\rangle $. 
Item
 5) follows in the same spirit from \eqref{eq: nonstandard bound h app}. 
 Items
 2,3) in Proposition \ref{prop: bounds on correlation functions} are addressed in Section \ref{sec: correlation functions mom cutoff}.

\subsubsection{Proof of the bound \eqref{eq: standard bound h app} in Lemma \ref{lem: bounds g} }
We write
\begin{align}
(\psi'_{s_2},\theta_{s}\psi_{s_1}) &=\int \e^{\i(|k_2|s_2-|k_1|s_1)}\hat\theta_{s}(k_1-k_2)
\overline{\hat\psi'(k_2)} \hat\psi(k_1)\d k_1\d k_2 \\
 &= \int_0^\infty \d\omega_1 \int_0^\infty \d\omega_2 
\e^{-\i(\omega_+s_-+\omega_-s_+)
}K(\omega_1,\omega_2)
\end{align}
where  $\omega_\pm:=\omega_1\pm\omega_2$,  $s_\pm=(s_1\pm s_2)/2$ and
$$
K(\omega_1,\omega_2)=(\omega_1\omega_2)^{d-1}\int_{S^{d-1}}\d\hat k_1\int_{S^{d-1}}\d\hat k_2 \, 
\overline{\hat\psi'(\omega_2\hat k_2)}\hat\psi(\omega_1\hat k_1)
s^{d}\zeta(s^2(\omega_1\hat k_1-\omega_2\hat k_2)^2)
$$
where by rotation invariance of $\theta$ we have written it as
  $\hat\theta(k)=\zeta(k^2)$ where $ \zeta$ satisfies
\beq
|\partial^n\zeta(x)|\leq C({n,N}) \langle x\rangle^{-N}  \label{appeq: bound on zeta}
\eeq
for all $n, N >0$.

\bigskip

\begin{lemma} \label{lem: prop k} There is a $\be>\al$ such that for $n=0,1,´2,3$, 
\beq  |\partial^{n}_{\omega_+}
K(\omega_1,\omega_2)|\leq C
(\omega_1^{-n}+\omega_2^{-n})
{(\omega_1\omega_2)^{d+\beta\over 2} \over
\langle\omega_1\rangle^{2}  \langle\omega_2\rangle^{2} }
\left({s(1/s+\omega_+)^{1-d}\over
\langle{_1\over^2}s^2\omega_-^2\rangle^{N/2}}
+
{s(1/s+|\omega_-|)^{1-d}\over
\langle {_1\over^2}s^2\omega_+^2\rangle^{N/2}}\right). \label{appeq: bounds on k}
\eeq
\end{lemma} 
\begin{proof}

Let $\hat k_1\cdot\hat k_2=:\cos\vartheta$ with $\vartheta \in [0, \pi]$. Then,
$$\omega_1^2+\omega_2^2-2
\omega_1\omega_2\cos\vartheta=\cos^2({_\vartheta\over^2}) \omega_-^2
+\sin^2({_\vartheta\over^2}) \omega_+^2 =: Z(\om_+,\om_-,\vartheta)
$$
and hence
$$
K(\omega_1,\omega_2)= s^{d} 
\int_{0}^\pi \d\vartheta\ (\sin\vartheta)^{d-2}
\zeta(s^2  Z(\om_+,\om_-,\vartheta)  )G(\omega_1,\omega_2,\vartheta)
$$
with
$$
G(\omega_1,\omega_2,\vartheta)=(\omega_1\omega_2)^{d-1}
\int_{S^{d-1}}\d\hat k_1\int_{S^{d-2}}\d\hat p  \, 
\overline{\hat\psi'(\omega_2\hat k_2)}\hat\psi(\omega_1\hat k_1)
$$
where $\hat k_2=\sin\vartheta\hat p+\cos\vartheta\hat k_1$ (and $\hat p\perp\hat k_1$).  Since $\psi,\psi' \in \frh_\al$, 
\beq
|\partial^{n_1}_{\omega_1}\partial^{n_2}_{\omega_2}
G(\omega_1,\omega_2,\vartheta)|\leq C
\prod_{i=1}^2\omega_i^{{{d+\beta}\over{2}}-n_i}
\langle \omega_i \rangle^{-2}, \label{appeq: bound on g one}
\eeq
uniformly in $\vartheta$, for some $\be >\al$.  This implies
$$
|\partial^{n}_{\omega_+}
G(\omega_1,\omega_2,\vartheta)|\leq C
(\omega_1^{-n}+\omega_2^{-n})
\prod_{i=1}^2\omega_i^{{{d+\beta}\over{2}}}
\langle \omega_i \rangle^{-2}.  \label{appeq: bound on g two}
$$
 From \eqref{appeq: bound on zeta} we deduce (we abbreviate $Z=Z(\om_+,\om_-,\vartheta) $)
$$
\str \partial_{\omega_+}^n \zeta(s^2 Z  ) \str \leq C\omega_+^{-n} 
\langle s^2Z \rangle^{-N}.\label{appeq: bound on Z derivatives}
$$
Combining the two previous inequalities, we get
\beq \label{appeq: bound on k}
|\partial^{n}_{\omega_+}
K(\omega_1,\omega_2)|\leq C
(\omega_1^{-n}+\omega_2^{-n})
{(\omega_1\omega_2)^{d+\beta\over 2} H(\omega_1,\omega_2)\over
\langle \omega_1 \rangle^{2} \langle \omega_2 \rangle^{2}}
\eeq
with
$$
H(\omega_1,\omega_2)= s^d\int_{0}^\pi \d\vartheta\ (\sin\vartheta)^{d-2}
\langle s^2Z \rangle^{-N}.
$$
For $\vartheta\in¬†[0,{\pi\over 2}]$ we have 
$$
Z
\geq {_1\over^2}(\omega_1-\omega_2)^2+{_1\over^4}\vartheta^2(\omega_1+\omega_2)^2
$$ 
and so
$$
\langle s^2Z \rangle^{-N}
\leq
\langle {_1\over^2}s^2\omega_-^2 \rangle^{-N/2}
\langle {_1\over^4}s^2\vartheta^2
\omega_+^2 \rangle^{-N/2}
$$
and for $\vartheta\in¬†[{\pi\over 2},{\pi}]$
$$
\langle s^2Z \rangle^{-N}
\leq
\langle {_1\over^2}s^2\omega_+^2 \rangle^{-N/2}
\langle {_1\over^4}s^2(\pi-\vartheta)^2
\omega_-^2 \rangle^{-N/2}.
$$
Since
$$
\int_{0}^{\pi/2} \d\vartheta\ (\sin\vartheta)^{d-2} \langle {_1\over^4}s^2\vartheta^2
\omega_+^2 \rangle^{-N/2}\leq C
(1+s\omega_+)^{1-d}
$$
and similarly for the integral over $[\pi/2,\pi]$ we get 
$$
H(\omega_1,\omega_2)
\leq
{s(1/s+\omega_+)^{1-d}\over
\langle {_1\over^2}s^2\omega_-^2 \rangle^{N/2}}
+
{s(1/s+\omega_-)^{1-d}\over
\langle{_1\over^2}s^2\omega_+^2\rangle^{N/2}}
$$
which yields the claim upon substitution in \eqref{appeq: bound on k}. 

\end{proof}

 Lemma \ref{lem: prop k}  implies  that the functions $ \omega^{1-\be}_+ \partial^{2}_{\omega_+}
K(\omega_1,\omega_2)$ and $\omega^{\be}_+  \partial^{3}_{\omega_+}
K(\omega_1,\omega_2)$ are integrable 
and  $\partial_{\omega_+}
K(\omega_1,\omega_2)$ vanishes if $\omega_1=0$ or $\omega_2=0$.
Hence
$$
(s_2-s_1)^2(\psi'_{s_2},\theta_{s}\psi_{s_1}) =-4\int_0^\infty \d\omega_1 \int_0^\infty \d\omega_2 
\e^{-\i(\omega_+s_-+\omega_-s_+)
}\partial_{\omega_+}^2K(\omega_1,\omega_2).
$$
By Lemma \ref{lem: decay} with $d=1$, we get, with $0<\ga<1$,  
$$
|(\psi'_{s_2},\theta_{s}\psi_{s_1}) |\leq C (s_2-s_1)^{-2-\gamma}\int_0^\infty \d\omega_1 \int_0^\infty \d\omega_2  \, \left(
 \str    \om^{-\ga}_+ \partial_{\omega_+}^{2}K(\omega_1,\omega_2) \str+   \str \om^{1-\ga}_+ \partial_{\omega_+}^{3}K(\omega_1,\omega_2) \str \right)
$$
provided that the right hand side is finite, which we prove now.
The contribution of the first term in the second parenthesis in \eqref{appeq: bounds on k} is dominated by
\beq
\int_0^\infty \d\omega_1 \int_0^\infty \d\omega_2 
{\omega_1^{{d+\beta\over 2}-2-\gamma} \omega_2^{d+\beta\over 2} \over
\langle\omega_1\rangle^{2}  \langle\omega_2\rangle^{2}}
{s(1/s+\omega_+)^{1-d}\over
\langle {_1\over^2}s^2\omega_-^2\rangle^{N/2}}   \quad +\quad   (\om_1 \leftrightarrow \om_2)   \label{eq: plug bounds on k into integral}
\eeq
where $ (\om_1 \leftrightarrow \om_2)$ stands for the same term but with $\om_1, \om_2$ interchanged. Since this term is treated in the same way, we drop it. 
Then, the $\omega_2$ integral in \eqref{eq: plug bounds on k into integral} gives the bound
$$
 \int_0^\infty \d\omega_2 
{ \omega_2^{d+\beta\over 2} \over
\langle\omega_2\rangle^{2} }
{s(1/s+\omega_+)^{1-d}\over
\langle {_1\over^2}s^2\omega_-^2\rangle^{N/2}}
\leq
C (\omega_1+1/s)^{{2+\beta-d\over 2}}.
$$
Indeed, for large $N$ the $
\langle {_1\over^2}s^2\omega_-^2\rangle^{-N/2}$ factor fixes $\omega_2=\omega_1+{\cal O}(1/s)$.  Therefore, we have
$$
\eqref{eq: plug bounds on k into integral} \leq C
\int_0^\infty {\d\omega_1} 
{\omega_1^{\beta-\gamma-1} \over
\langle\omega_1\rangle^{2}}
({\omega_1\over \omega_1+1/s})^{{d-2-\beta\over 2}}
\leq C
$$
(uniformly in $s$)  if $\gamma<\beta$ because $d \geq 3$. 

The second term between brackets in \eqref{appeq: bounds on k} is bounded uniformly in $s$ for $\ga<\be$, so the bound \eqref{eq: standard bound h app} is proven.

\subsubsection{Proof of the bound \eqref{eq: nonstandard bound h app} in Lemma \ref{lem: bounds g}}
For $\psi \in \frh_\al$, we write
$$
\psi_s(x)=\int_{S^{d-1}} \d\hat k\int_0^\infty \d \om \,  \om^{d-1}\e^{-\i \om (s+\hat k\cdot x)}\hat\psi(\om\hat k), 
$$
and, by Lemma \ref{lem: decay},  for some $\be>\al$,
$$
\str\psi_{s}(x)\str
\leq   C 
\int_{S^{d-1}} \d \hat k\  \langle s+\hat k\cdot x\rangle^{-d+{1-\beta\over 2}}.
$$
Therefore, for $\psi,\psi' \in \frh_\al$ and $m'_\theta$ such that $r_\theta<m'_\theta <m_\theta$
$$
\str (\psi_{s_1}, \theta_{s}(x)\psi'_{s_2}) \str \leq 
Cs^d \langle s_1-{ m'_\theta s }\rangle^{-d+{1-\beta\over 2}}
  \langle s_2-m'_\theta s\rangle^{-d+{1-\beta\over 2}}
$$
since $\str \hat k\cdot x \str \leq {m'_\theta s}$ on the support of $\theta_{s}$.  Combined with \eqref{eq: standard bound h app}, this yields \eqref{eq: nonstandard bound h app}.

\subsection{Momentum cutoff} \label{sec: correlation functions mom cutoff}

We treat the case where $b= \theta(k/\delta)$, i.e.\ in momentum space. 
\begin{lemma}\label{lem: app momentum cutoff} Let $\psi,\psi' \in \frh_{\al}$, then there is a $\be>\al$  such that for any $\ga,\ga' \geq 0$ with $\ga+\ga' \leq \be$;  
\beq \label{eq: integral mom cutoff}
 \str (\psi_{s_1}, \theta(k/\delta) \psi'_{s_2}) \str   \leq  C \delta^{\ga'}   \langle s_2-s_1 \rangle^{-(2+\ga)}. 
\eeq
\end{lemma}
\begin{proof}
Set 
\beq
g(k) :=     \str k \str^{d-1} \theta(\str k \str/\delta)\overline{\hat \psi(k)}  \hat \psi'(k)
\eeq
We bound, for $n=0,1,2,3$, 
\beq \label{eq: bounds on f}
\str\partial^{n}_{\om}   g(\om\hat k ) \str \leq  C \om^{1+\be-n} \nu(\om/\delta), \qquad \text{with}\,\,  \nu(\om) =\sum_{j=0}^2 \str\partial^j\theta(\om) \str
\eeq
where we wrote $\theta (\str k \str)=\theta(k)$ because of spherical symmetry, and we used that  $\nu(\om) =0$ for $\om > 1$.  Since $\partial_\om g, \partial^2_\om g$ are integrable in $\om$ by the bounds \eqref{eq: bounds on f}, the left hand side of \eqref{eq: integral mom cutoff} is bounded by 
\beq
C \langle s_2-s_1 \rangle^{-2}  \int_{S^{d-1}} \d \hat k  \str f_{\hat k}(s_2-s_1)\str,
\eeq
with  $f_{\hat k}(\cdot) $ the inverse Fourier transform of $\hat f_{\hat k}(\om) := 1_{\om >0} \partial^2_\om g(\om\hat k )$.  
Furthermore, by  \eqref{eq: bounds on f}
\beq
\norm \om^{1-\ga} \hat f_{\hat k} \norm_{L^1(\d \om)}  + \norm  \om^{-\ga}   \partial_\om \hat f_{\hat k} \norm_{L^1(\d \om)}  \leq C \delta^{\ga'},
\eeq
uniformly in $\hat k$. 
We can now apply Lemma \ref{lem: decay} to the function $ f_{\hat k}$ and we get the required bound. 
\end{proof}

As described following Lemma \ref{lem: bounds g}, the above  
Lemma \ref{lem: app momentum cutoff} yields items 4) and 2)  of Proposition \ref{prop: bounds on correlation functions}.

\bibliographystyle{abbrv}
\bibliography{mylibrary11}

\end{document}